\documentclass[a4paper]{llncs}

\usepackage[utf8]{inputenc}

\usepackage{amsfonts}
\usepackage{amsmath}

\usepackage{mathtools}

\def\Xrightarrow[#1]{\xrightarrow{#1}\!\!^*}
\def\XRightarrow[#1]{\xRightarrow{#1}\!\!^*}

\usepackage{shuffle}

\usepackage{tikz}
\usepackage[justification=centering]{caption}
\usetikzlibrary{automata, positioning, chains, fit, shapes, patterns, matrix}
\usepackage{graphicx}
\usepackage{xcolor}

\usepackage{cite}
\usepackage[autostyle=true]{csquotes}

\title{Static Analysis of Multithreaded Recursive Programs Communicating via Rendez-vous\thanks{This work was partially funded by the FUI project FREENIVI.}}
\author{Adrien Pommellet\inst{1} \and Tayssir Touili\inst{2}}
\institute{LIPN and Université Paris-Diderot, France
	\and LIPN, CNRS, and Université Paris 13, France}

\begin{document}

\maketitle

\begin{abstract}
	We present in this paper a generic framework for the analysis of multi-threaded programs with recursive procedure calls, synchronisation by rendez-vous between parallel threads, and dynamic creation of new threads. To this end, we consider a model called \emph{Synchronized Dynamic Pushdown Networks} (SDPNs) that can be seen as a network of pushdown processes executing synchronized transitions, spawning new pushdown processes, and performing internal pushdown actions. The reachability problem for this model is unfortunately undecidable. Therefore, we tackle this problem by introducing an abstraction framework based on Kleene algebras in order to compute an abstraction of the execution paths between two regular sets of configurations. We combine an automata theoretic saturation procedure with constraint solving in a finite domain. We then apply this framework to an iterative abstraction refinement scheme, using multiple abstractions of increasing complexity and precision.
	
\end{abstract}

The use of parallel programs has grown in popularity in the past fifteen years, but these remain nonetheless fickle and vulnerable to specific issues such as race conditions or deadlocks. Static analysis methods for this class of programs remain therefore more relevant than ever.

\emph{Pushdown Systems} (PDSs) are a natural model for programs with sequential, recursive procedure calls, as shown by Esparza et al. in \cite{EHRS-cav00}. Thus, networks of pushdown systems can be used to model multithreaded programs, where each PDS in the network models a sequential component of the whole program. In this context, \emph{Dynamic Pushdown Networks} (DPNs) were introduced by Bouajjani et al. in \cite{BMT-concur05}.

Intuitively, this class of automata consists of a network of pushdown systems running independently in parallel. Each member of a DPN can, after a transition, spawn a new PDS which is then introduced as a new member of the network. Thus, DPNs can be used to represent a network of threads where each thread can recursively call procedures, perform internal actions, or spawn a new thread.

However, this model cannot represent synchronization between different \linebreak threads or parallel components. In order to handle communication in multithreaded programs, Bouajjani et al. introduced in \cite{BET-popl03} \emph{communicating pushdown systems} (CPDSs), a model which consists of a tuple of pushdown systems synchronized by rendez-vous on execution paths. However, CPDSs have a constant number of processes and cannot therefore handle dynamic creation of new threads.

Hence, we introduce a more accurate model, namely, \emph{synchronized dynamic pushdown networks} (SDPNs) that combines DPNs with CPDSs in order to handle dynamic thread creation and communication at the same time.

A SDPN can be seen as a DPN where PDS processes can synchronize via rendez-vous by sending and receiving messages. In a SDPN, pushdown processes can apply internal actions labeled by a letter $\tau$ without synchronization, just like a DPN, but can also synchronize through channels.

To do so, we represent each channel by a pair of letters, as an example $a$ and $\overline{a}$, that can be used to label transitions. If one thread can execute an action labeled with a signal $a$, and another thread another action labeled with $\overline{a}$, then both threads can synchronize and execute their respective transitions simultaneously, in a single step labeled by $\tau$.

We consider the reachability problem for SDPNs, that is, finding if a critical configuration can be reached from the set of starting configurations of the program. An equivalent problem is to compute the set $Paths ( \mathcal{C}, \mathcal{C'} )$ of execution paths leading from a configuration in $\mathcal{C}$ to a configuration in $C'$ and check if it is empty. This problem unfortunately remains undecidable for synchronized pushdown systems, as proven by Ramalingam in \cite{R-acm00}.

Therefore, the set of execution paths $Paths ( C, C' )$ cannot be computed in an exact manner. To overcome this problem, we proceed in a manner similar to the method outlined in \cite{BET-popl03}: our approach is based on the computation of an abstraction $\alpha ( Paths ( \mathcal{C}, \mathcal{C'} ) )$ of the execution paths language. To this aim, we propose techniques based on:
\begin{itemize}
	\item the representation of regular sets of configurations of SDPNs with finite word automata;
	
	\item the use of these automata to determine a set of constraints whose least fixpoint characterizes the set of execution paths of the program; to compute this set of constraints, (1) we consider a relaxed semantics on SDPNs that allows partially synchronized runs, (2) we abstract sets of execution paths as functions in a Kleene algebra, instead of simple elements of the abstract domain, and (3) we use a shuffle product on abstract path expressions to represent the interleaving and potential synchronization of parallel executions;
	
	\item the resolution of this set of constraints in an abstract domain; we consider in particular the case where the abstract domain is finite; the set of constraints can then be solved using an iterative fixpoint computation.
\end{itemize}
Note that the main contribution of our approach with regards to the methods outlined \cite{BET-popl03, T-vissas05} is the introduction of functions to represent sets of abstracted path expressions and the use of a shuffle product to model the interleaving of threads. The abstraction framework as defined in these papers cannot be applied to SDPNs due to the presence of dynamic thread creation, hence, the need for functions and shuffling.

We can then apply this over-approximation framework for the reachability problem to an \emph{iterative abstraction refinement} scheme inspired by the work of Chaki et al. in \cite{CCKRT-etaps06}. The idea is the following: (1) we do a reachability analysis of the program, using a finite domain abstraction of order $n$ in our over-approximation framework; if the target set of configurations is not reachable by the abstract paths, it is not reachable by actual execution paths either; otherwise, we obtain a counter-example; (2) we check if the counter-example can be matched to an actual execution of the program; (3) if it does, then we have shown that the target set of configurations is actually reachable; (4) otherwise, we refine our abstraction and use instead a finite domain abstraction of order $n + 1$ in step (2). This scheme is then used to prove that a Windows driver first presented in \cite{QW-pldi04} can reach an erroneous configuration, using an abstraction of the original program. An updated version of this driver is then shown to be error-free.

This paper is a full, corrected version of \cite{PT-aplas17}.

\smallskip
\noindent
{\bf Paper outline.}
In Section $1$ of this paper, we define \emph{synchronized dynamic pushdown networks} (SDPNs). We study in Section $2$ the reachability problem for this class of automata. We introduce in Section $3$ an automata-theoretic representation of sets of paths, and prove in Section $4$ that the set of execution paths between two sets of configurations $C$ and $C'$ of a SDPN is the least solution of a set constraints. Since we can't solve these constraints, we present in Section $5$ an abstraction framework for paths based on Kleene algebras. In Section $6$, we apply this framework to over-approximate the reachability problem. In Section $7$, we present a \emph{iterative abstraction refinement} scheme that relies on our abstraction framework and apply it to a model of an actual program in section $8$. Finally, we describe the related work in Section $9$ and show our conclusion in Section $10$.

\section[Synchronized dynamic pushdown networks]{Synchronized dynamic pushdown networks}

\subsection{Dynamic pushdown networks}

We briefly introduce this class of automata:

\begin{definition}[Bouajjani et al. \cite{BMT-concur05}]
	A \emph{dynamic pushdown network} (DPN) is a triplet $M = ( P, \Gamma, \Delta )$ where $P$ is a finite set of control states, $\Gamma$ a finite stack alphabet disjoint from $P$, and $\Delta \subseteq ( P \Gamma \times P \Gamma^* ) \cup ( P \Gamma \times P \Gamma^* P \Gamma^* )$ a finite set of transition rules featuring:
	\begin{itemize}
		\item simple pushdown operations in $( P \Gamma \times P \Gamma^* )$ of the form $p \gamma \rightarrow p' w$;
		\item thread spawns in $( P \Gamma \times P \Gamma^* P \Gamma^* )$ of the form $p \gamma \rightarrow p_1 w_1 p_2 w_2$.
	\end{itemize}
\end{definition}

Let $Conf_M = \left(P \Gamma^* \right)^*$ be the set of configurations of a DPN $M$. A configuration $p_1 w_1 \ldots p_n w_n$ represents a network of $n$ processes where the $i$-th process is in control point $p_i$ and has stack content $w_i$, as shown in Figure \ref{fig:4-DPNconf} where a single word in $Conf_M$ is used to represent the state of three PDSs in a network.

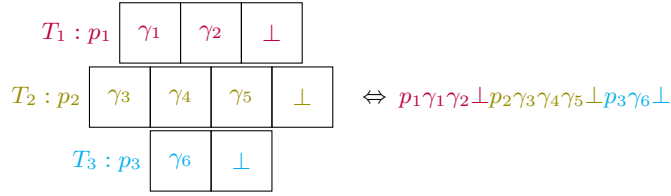
\begin{figure}
	\centering
	\begin{minipage}{.4\linewidth}
		\centering
		\begin{tikzpicture}
		% Paramètres		
		\edef\sizetape{0.8cm}
		\tikzstyle{tmtape}=[draw,minimum size=\sizetape]
		
		% Ruban
		\begin{scope}[start chain=1 going right,node distance=0mm, outer sep=0mm]
		\node [on chain=1,tmtape,draw=none]{$\textcolor{purple}{T_1 : p_1}$};
		\node [on chain=1,tmtape]{$\textcolor{purple}{\gamma_1}$};
		\node [on chain=1,tmtape]{$\textcolor{purple}{\gamma_2}$};
		\node [on chain=1,tmtape]{$\textcolor{purple}{\bot}$};
		\end{scope}
		\end{tikzpicture}
		\begin{tikzpicture}
		% Paramètres		
		\edef\sizetape{0.8cm}
		\tikzstyle{tmtape}=[draw,minimum size=\sizetape]
		
		% Ruban
		\begin{scope}[start chain=1 going right,node distance=0mm, outer sep=0mm]
		\node [on chain=1,tmtape,draw=none]{$\textcolor{olive}{T_2 : p_2}$};
		\node [on chain=1,tmtape]{$\textcolor{olive}{\gamma_3}$};
		\node [on chain=1,tmtape]{$\textcolor{olive}{\gamma_4}$};
		\node [on chain=1,tmtape]{$\textcolor{olive}{\gamma_5}$};
		\node [on chain=1,tmtape]{$\textcolor{olive}{\bot}$};
		\end{scope}
		\end{tikzpicture}
		\begin{tikzpicture}
		% Paramètres		
		\edef\sizetape{0.8cm}
		\tikzstyle{tmtape}=[draw,minimum size=\sizetape]
		
		% Ruban
		\begin{scope}[start chain=1 going right,node distance=0mm, outer sep=0mm]
		\node [on chain=1,tmtape,draw=none]{$\textcolor{cyan}{T_3 : p_3}$};
		\node [on chain=1,tmtape]{$\textcolor{cyan}{\gamma_6}$};
		\node [on chain=1,tmtape]{$\textcolor{cyan}{\bot}$};
		\end{scope}
		\end{tikzpicture}
	\end{minipage}
	\begin{minipage}{.5\linewidth}
		$\Leftrightarrow \medspace \textcolor{purple}{p_1 \gamma_1 \gamma_2 \bot} \textcolor{olive}{p_2 \gamma_3 \gamma_4 \gamma_5 \bot} \textcolor{cyan}{p_3 \gamma_6 \bot}$
	\end{minipage}
	\captionof{figure}{Representing configurations of a DPN.}
	\label{fig:4-DPNconf}
\end{figure}

We define an immediate successor relation $\rightarrow_{M}$ on $Conf_M$ according to the following semantics:
\begin{itemize}
	\item if $p \gamma \rightarrow p' w$ in $\Delta$, then $\forall u, v \in Conf_M$, $\forall w' \in \Gamma^*$, $u p \gamma w' v \rightarrow_{M} u p' w w' v$; a thread applies a pushdown operation on its own stack, as shown in Figure \ref{fig:4-DPNpushpop};
	
	\item if $p \gamma \rightarrow p_1 w_1 p_2 w_2$ in $\Delta$, then $\forall u, v \in Conf_M$, $\forall w' \in \Gamma^*$, $u p \gamma w' v \rightarrow_{M} u p_1 w_1 p_2 w_2 w' v$; a thread spawns a new son with its own stack and control state, as shown in Figure \ref{fig:4-DPNspawn}.
\end{itemize}

\begin{figure}
	\centering
	\begin{minipage}{.45\linewidth}
		\centering
		\begin{tikzpicture}
		% Paramètres		
		\edef\sizetape{0.8cm}
		\tikzstyle{tmtape}=[draw,minimum size=\sizetape]
		
		% Ruban
		\begin{scope}[start chain=1 going right,node distance=0mm, outer sep=0mm]
		\node [on chain=1,tmtape,draw=none]{$T_1 : p_1$};
		\node [on chain=1,tmtape]{$\gamma_1$};
		\node [on chain=1,tmtape]{$\gamma_2$};
		\end{scope}
		\end{tikzpicture}
		\begin{tikzpicture}
		% Paramètres		
		\edef\sizetape{0.8cm}
		\tikzstyle{tmtape}=[draw,minimum size=\sizetape]
		
		% Ruban
		\begin{scope}[start chain=1 going right,node distance=0mm, outer sep=0mm]
		\node [on chain=1,tmtape,draw=none]{$\textcolor{olive}{T_2 : p_2}$};
		\node [on chain=1,tmtape]{$\textcolor{olive}{\gamma_3}$};
		\node [on chain=1,tmtape]{$\gamma_4$};
		\node [on chain=1,tmtape]{$\gamma_5$};
		\end{scope}
		\end{tikzpicture}
		\begin{tikzpicture}
		% Paramètres		
		\edef\sizetape{0.8cm}
		\tikzstyle{tmtape}=[draw,minimum size=\sizetape]
		
		% Ruban
		\begin{scope}[start chain=1 going right,node distance=0mm, outer sep=0mm]
		\node [on chain=1,tmtape,draw=none]{$\textcolor{purple}{T_3 : p_3}$};
		\node [on chain=1,tmtape]{$\gamma_6$};
		\end{scope}
		\end{tikzpicture}
	\end{minipage}
	\hspace{.01\linewidth}
	\begin{minipage}{.45\linewidth}
		\centering
		\begin{tikzpicture}
		% Paramètres		
		\edef\sizetape{0.8cm}
		\tikzstyle{tmtape}=[draw,minimum size=\sizetape]
		
		% Ruban
		\begin{scope}[start chain=1 going right,node distance=0mm, outer sep=0mm]
		\node [on chain=1,tmtape,draw=none]{$T_1 : p_1$};
		\node [on chain=1,tmtape]{$\gamma_1$};
		\node [on chain=1,tmtape]{$\gamma_2$};
		\end{scope}
		\end{tikzpicture}
		\begin{tikzpicture}
		% Paramètres		
		\edef\sizetape{0.8cm}
		\tikzstyle{tmtape}=[draw,minimum size=\sizetape]
		
		% Ruban
		\begin{scope}[start chain=1 going right,node distance=0mm, outer sep=0mm]
		\node [on chain=1,tmtape,draw=none]{$\textcolor{olive}{T_2 : p'_2}$};
		\node [on chain=1,tmtape]{$\gamma_4$};
		\node [on chain=1,tmtape]{$\gamma_5$};
		\end{scope}
		\end{tikzpicture}
		\begin{tikzpicture}
		% Paramètres		
		\edef\sizetape{0.8cm}
		\tikzstyle{tmtape}=[draw,minimum size=\sizetape]
		
		% Ruban
		\begin{scope}[start chain=1 going right,node distance=0mm, outer sep=0mm]
		\node [on chain=1,tmtape,draw=none]{$\textcolor{purple}{T_3 : p'_3}$};
		\node [on chain=1,tmtape]{$\textcolor{purple}{\gamma_7}$};
		\node [on chain=1,tmtape]{$\gamma_6$};
		\end{scope}
		\end{tikzpicture}
	\end{minipage}
	\captionof{figure}{A DPN with 3 threads after a \textcolor{olive}{pop} from $T_2$ and a \textcolor{purple}{push} on $T_3$.}
	\label{fig:4-DPNpushpop}
\end{figure}
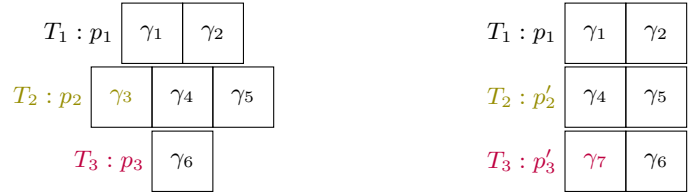

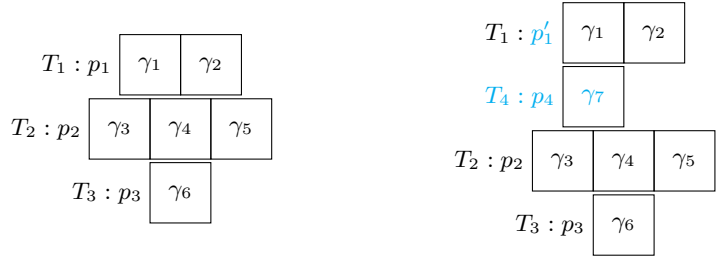
\begin{figure}
	\centering
	\begin{minipage}{.45\linewidth}
		\centering
		\begin{tikzpicture}
		% Paramètres		
		\edef\sizetape{0.8cm}
		\tikzstyle{tmtape}=[draw,minimum size=\sizetape]
		
		% Ruban
		\begin{scope}[start chain=1 going right,node distance=0mm, outer sep=0mm]
		\node [on chain=1,tmtape,draw=none]{$T_1 : p_1$};
		\node [on chain=1,tmtape]{$\gamma_1$};
		\node [on chain=1,tmtape]{$\gamma_2$};
		\end{scope}
		\end{tikzpicture}
		\begin{tikzpicture}
		% Paramètres		
		\edef\sizetape{0.8cm}
		\tikzstyle{tmtape}=[draw,minimum size=\sizetape]
		
		% Ruban
		\begin{scope}[start chain=1 going right,node distance=0mm, outer sep=0mm]
		\node [on chain=1,tmtape,draw=none]{$T_2 : p_2$};
		\node [on chain=1,tmtape]{$\gamma_3$};
		\node [on chain=1,tmtape]{$\gamma_4$};
		\node [on chain=1,tmtape]{$\gamma_5$};
		\end{scope}
		\end{tikzpicture}
		\begin{tikzpicture}
		% Paramètres		
		\edef\sizetape{0.8cm}
		\tikzstyle{tmtape}=[draw,minimum size=\sizetape]
		
		% Ruban
		\begin{scope}[start chain=1 going right,node distance=0mm, outer sep=0mm]
		\node [on chain=1,tmtape,draw=none]{$T_3 : p_3$};
		\node [on chain=1,tmtape]{$\gamma_6$};
		\end{scope}
		\end{tikzpicture}
	\end{minipage}
	\hspace{.01\linewidth}
	\begin{minipage}{.45\linewidth}
		\centering
		\begin{tikzpicture}
		% Paramètres		
		\edef\sizetape{0.8cm}
		\tikzstyle{tmtape}=[draw,minimum size=\sizetape]
		
		% Ruban
		\begin{scope}[start chain=1 going right,node distance=0mm, outer sep=0mm]
		\node [on chain=1,tmtape,draw=none]{$T_1 : \textcolor{cyan}{p'_1}$};
		\node [on chain=1,tmtape]{$\gamma_1$};
		\node [on chain=1,tmtape]{$\gamma_2$};
		\end{scope}
		\end{tikzpicture}
		\begin{tikzpicture}
		% Paramètres		
		\edef\sizetape{0.8cm}
		\tikzstyle{tmtape}=[draw,minimum size=\sizetape]
		
		% Ruban
		\begin{scope}[start chain=1 going right,node distance=0mm, outer sep=0mm]
		\node [on chain=1,tmtape,draw=none]{$\textcolor{cyan}{T_4 : p_4}$};
		\node [on chain=1,tmtape]{$\textcolor{cyan}{\gamma_7}$};
		\node [on chain=1,tmtape,draw=none]{};
		\end{scope}
		\end{tikzpicture}
		\begin{tikzpicture}
		% Paramètres		
		\edef\sizetape{0.8cm}
		\tikzstyle{tmtape}=[draw,minimum size=\sizetape]
		
		% Ruban
		\begin{scope}[start chain=1 going right,node distance=0mm, outer sep=0mm]
		\node [on chain=1,tmtape,draw=none]{$T_2 : p_2$};
		\node [on chain=1,tmtape]{$\gamma_3$};
		\node [on chain=1,tmtape]{$\gamma_4$};
		\node [on chain=1,tmtape]{$\gamma_5$};
		\end{scope}
		\end{tikzpicture}
		\begin{tikzpicture}
		% Paramètres		
		\edef\sizetape{0.8cm}
		\tikzstyle{tmtape}=[draw,minimum size=\sizetape]
		
		% Ruban
		\begin{scope}[start chain=1 going right,node distance=0mm, outer sep=0mm]
		\node [on chain=1,tmtape,draw=none]{$T_3 : p_3$};
		\node [on chain=1,tmtape]{$\gamma_6$};
		\end{scope}
		\end{tikzpicture}
	\end{minipage}
	\captionof{figure}{A DPN with 3 threads after thread $T_1$ \textcolor{cyan}{spawns} a new thread $T_4$.}
	\label{fig:4-DPNspawn}
\end{figure}

Let $\rightarrow^*_{M}$ be the transitive and reflexive closure of this relation. Given a set $\mathcal{C} \subseteq Conf_M$ of configurations, we introduce its set of \emph{predecessors} $pre^* ( M, \mathcal{C} ) = \lbrace c \in Conf_M \mid \exists c' \in \mathcal{C}, c \Rightarrow_\mathcal{M} c' \rbrace$. If $\mathcal{C}$ is regular, this set can be effectively computed:

\begin{theorem}[Bouajjani et al. \cite{BMT-concur05}]
	\label{thm_pre_DPN_regular}
	Given a DPN $M$ and a regular set of configurations $\mathcal{C} \subseteq Conf_M$, the set $pre^* ( M, \mathcal{C} )$ of predecessors is regular.
\end{theorem}

The saturation algorithm used to compute $pre^* ( M, \mathcal{C} )$ is detailled in Section \ref{compute_preDPN}.

\subsection{The model and its semantics}

We introduce a new model:

\begin{definition}
	A \emph{synchronized dynamic pushdown Network} (SDPN) is a tuple $M = ( Act, P, \Gamma, \Delta )$ where $Act$ is a finite set of actions, $P$ a finite set of control states, $\Gamma$ a finite stack alphabet disjoint from $P$, and $\Delta \subseteq ( P \Gamma \times Act \times P \Gamma^* ) \cup ( P \Gamma \times Act \times P \Gamma^* P \Gamma^* )$ a finite set of transition rules.
\end{definition}

If $( p \gamma, a, w ) \in \Delta$, $p \in P$, $\gamma \in \Gamma$, $a \in Act$, and $w \in P \Gamma^* \cup P \Gamma^* P \Gamma^*$, we write that $p \gamma \xrightarrow{a} w \in \Delta$. There are two types of transition rules in a SDPN:
\begin{itemize}
	\item rules of the form $p \gamma \xrightarrow{a} p' w$ in $P \Gamma \times Act \times P$ allow a pushdown process in the network to pop a symbol $\gamma$ from its stack, push a word $w$, then move from state $p$ to $p'$; these rules are standard pushdown rules and model a thread calling or ending procedures while moving through its control flow;
	
	\item rules of the form $p \gamma \xrightarrow{a} p'' w' p' w$ in $P \Gamma \times Act \times P \Gamma^* P \Gamma^*$ allow a pushdown process in the network to pop a symbol $\gamma$ from its stack, push a word $w$, move from state $p$ to $p'$, then spawn a new pushdown process in state $p''$ and with initial stack $w'$; these rules model dynamic creation of new threads.
\end{itemize}

We assume that the set $Act$ contains a letter $\tau$ that represents internal or synchronized actions, and that other letters in $Lab = Act \setminus \{ \tau \}$ model synchronization signals. Moreover, to each synchronization signal $a$ in $Lab$, we can match an unique co-action $\overline{a} \in Lab$, such that $\overline{\overline{a}} = a$.

We introduce the set $Conf_M = \left(P \Gamma^* \right)^*$ of configurations of a SDPN $M$. In a manner similar to DPNs, a configuration $p_1 w_1 \ldots p_n w_n$ represents a network of $n$ processes where the $i$-th process is in control point $p_i$ and has stack content $w_i$.

\subsubsection{The strict semantics.}

We will model synchronization between threads as a form of \emph{communication by rendez-vous}: two pushdown processes can synchronize if one performs a transition labelled with $a$ and the other, a transition labelled with $\overline{a}$. Intuitively, one thread sends a signal over a channel and the other thread waits for a signal to be received along the same channel.

To this end, we define a strict transition relation $\dashrightarrow_{M}$ on configurations of $M$ according to the following \emph{strict semantics}:
\begin{description}
	\item[(1)] given a symbol $a \in Act$, two rules $p \gamma \xrightarrow{a} w_1$ and $p' \gamma' \xrightarrow{\overline{a}} w'_1$ in $\Delta$, and two configurations $u = u_1 p \gamma u_2 p' \gamma' u_3$ and $v=u_1 w_1 u_2 w'_1 u_3$ of $M$, we have $u \dashrightarrow_{M} v$; two synchronized processes perform a simultaneous action, as shown in Figure \ref{fig:4-strict1};
	
	\item[(2)] given a rule $p \gamma \xrightarrow{\tau} w_1$ in $\Delta$ and two configurations $u = u_1 p \gamma u_2$ and $v = u_1 w_1 u_2$ of $M$, we have $u \dashrightarrow_{M} v$; a process performs an internal action, as shown in Figure \ref{fig:4-strict2}.
\end{description}

\begin{figure}
	\centering
	\begin{tabular}{ccccc}
		$u_1$ & \textcolor{teal}{$p \gamma w$} & $u_2$ & \textcolor{olive}{$p' \gamma' w'$} & $u_3$ \\
		& $\downarrow a$ & & $\downarrow \overline{a}$ & \\
		$u_1$ & \textcolor{teal}{$w_1 w$} & $u_2$ & \textcolor{olive}{$w'_1 w'$} & $u_3$
	\end{tabular}
	\captionof{figure}{Semantics of synchronized actions.}
	\label{fig:4-strict1}
\end{figure}

\begin{figure}
	\centering
	\begin{tabular}{ccc}
		$u_1$ & \textcolor{purple}{$p \gamma w$} & $u_2$ \\
		& $\downarrow \tau$ & \\
		$u_1$ & \textcolor{purple}{$w_1 w$} & $u_2$
	\end{tabular}
	\captionof{figure}{Semantics of internal actions.}
	\label{fig:4-strict2}
\end{figure}

We say that $v$ is reachable from $u$ with regards to the strict semantics if $u \dashrightarrow^*_{M} v$, where $\dashrightarrow^*_{M}$ stands for the transitive closure of $\dashrightarrow_{M}$.

The strict semantics accurately model communication by rendez-vous. However, for technical matters, we also need to consider a relaxed semantics for SDPNs.

\subsubsection{The relaxed semantics.}

The \emph{relaxed semantics} on SDPNs allows partially synchronized executions on a SDPN: a process can perform a transition labelled with $a \in Lab$ even if doesn't synchronize with a matching process executing a transition labelled with $\overline{a}$.

We therefore introduce a relaxed transition relation $\rightarrow_{M}$ labelled in $Act$ on configurations of $M$:
\begin{description}
	\item[(1) \& (2)] given two configurations $u$ and $v$ of $M$, $u \dashrightarrow_{M} v$ if and only if $u \xrightarrow{\tau}_M v$; $\rightarrow_{M}$ features rules \textbf{(1)} and \textbf{(2)} of $\dashrightarrow_{M}$;
	
	\item[(3)] given a rule $p \gamma \xrightarrow{a} w_1$ in $\Delta$, a word $w_1 \in ( P \Gamma^* ) \cup ( P \Gamma^* )^2$, and two configurations $u = u_1 p \gamma u_2$ and $v = u_1 w_1 u_2$ of $M$, we have $u \xrightarrow{a}_M v$; a process performs an action but does not synchronize, as shown in Figure \ref{fig:4-relax}.
\end{description}
The restriction of the relaxed semantics to rules \textbf{(2)} and \textbf{(3)} yields the DPN semantics, as defined by Bouajjani et al. in \cite{BMT-concur05}.

\begin{figure}
	\centering
	\begin{tabular}{ccc}
		$u_1$ & \textcolor{purple}{$p \gamma w$} & $u_2$ \\
		& $\downarrow a$ & \\
		$u_1$ & \textcolor{purple}{$w_1 w$} & $u_2$
	\end{tabular}
	\captionof{figure}{Semantics of unsynchronized actions.}
	\label{fig:4-relax}
\end{figure}

For a given word $\sigma = a_1 \ldots a_n \in Act^*$ and two configurations $c$, $c'$ of $M$, we write that $c \Xrightarrow[\sigma]_{M} c'$ if there are $n$ configurations $c_1, \ldots, c_n$ of $M$ such that $c \xrightarrow{a_1}_M c_1 \xrightarrow{a_2}_M c_2 \ldots \xrightarrow{a_n}_M c_n$ and $c_n = c'$. We then say that $c'$ is reachable from $c$ with regards to the relaxed semantics. For a given set of configurations $C$, we introduce $pre^* ( M, C ) = \{ c' \mid \exists c \in C, \exists w \in \Gamma^*, c'\Xrightarrow[w]_M c \}$.

For two subsets $C$ and $C'$ of $Conf_M$, we define the set $Paths_M ( C, C' ) = \lbrace \sigma \in Act^* \mid \exists c \in C, \exists c' \in C', c \Xrightarrow[\sigma]_{M} c' \rbrace$ of all execution paths from $C$ to $C'$, including paths with non-synchronized actions labelled in $Lab$.

\subsection{From a program to a SDPN model}
\label{sdpn_program_model}

We can assume that the program is given by a \emph{control flow graph}, whose nodes represent control points of threads or procedures and whose edges are labelled by statements. These statements can be variable assignments, procedure calls or returns, spawns of new threads, or communications between threads through unidirectional point-to point channels, where a thread sends a value $x$ through a channel $c$ and another thread waits for this value then assigns it to a variable $y$.

Without loss of generality, we assume that threads share no global variables and instead can only synchronize through channels. We distinguish local variables that belong to a single procedure from thread-local variables that can be accessed by any procedure called by a given instance of a thread. We also consider that both local and global variables may only take a finite number of values.

Given a control flow graph, we define a corresponding SDPN. The set of states $P$ is the set of all possible valuations of thread-local variables. The stack alphabet $\Gamma$ is the set of all pairs $( n, l )$ where $n$ is a node of the flow graph and $l$ is a valuation of the local variables of the current procedure.

Channels can be used to send and receive values. For each channel $c$ and value $x$ that can be sent through $c$, a label $( c!, x ) $ and its co-action $( c?, x ) = \overline{( c!, x )}$ belong to $Act$. The internal action $\tau$ belongs to $Act$ as well.

For each statement $s$ labelling an edge of the flow graph between nodes $n_1$ and $n_2$, we introduce the following transition rules in the corresponding SDPN, where $g_1$ and $g_2$ (resp. $l_1$ and $l_2$) are the valuations of thread-local (resp. procedure-local) variables before and after the execution of the statement:
\begin{itemize}
	\item if $s$ is an assignment, rules of the form $g_1 ( n_1, l_1 ) \xrightarrow{\tau} g_2 ( n_2, l_2 )$ represent $s$; assigning new values to variables in $g_1$ and $l_1$ results in new valuations $g_2$ and $l_2$;
	
	\item if $s$ is a procedure call, rules of the form $g_1 ( n_1, l_1 ) \xrightarrow{\tau} g_2 ( f_0, l_0 ) ( n_2, l_2 )$ represent $s$, where $f_0$ is the starting node of the called procedure and $l_0$ the initial valuation of its local variables;
	
	\item if $s$ is a procedure return, it is represented by rules of the form $g_1 ( n_1, l_1 ) \xrightarrow{\tau} g_2 \varepsilon$; we simulate returns of values by introducing an additional thread-local variable and assigning the return value to it in the valuation $g_2$;
	
	\item if $s$ is a thread spawn, it is represented by rules of the form $g_1 ( n_1, l_1 ) \xrightarrow{\tau} g_0 ( n_0, l_0 ) g_2 ( n_2, l_2 )$, where $g_0$ and $l_0$ are respectively the initial valuations of the thread-local and procedure-local variables of the new thread, and $n_0$ its starting node;
	
	\item if $s$ is an assignment of a value $x$ carried through a channel $c$ to a variable $y$, it is represented by rules of the form $g_1 ( n_1, l_1 ) \xrightarrow{( c?, x )} g_2 ( n_2, l_2 )$ where $g_1$ and $g_2$ (resp. $l_1$ and $l_2$) are such that assigning the value $x$ to the variable $y$ in $g_1$ (resp. $l_1$) results in the new valuations $g_2$ (resp. $l_2$);
	
	\item if $s$ is an output through a channel $c$ of the value $x$ of a variable $y$, it is represented by rules of the form $g_1 ( n_1, l_1 ) \xrightarrow{( c!, x )} g_2 ( n_2, l_2 )$ such that the variable $y$ has value $x$ in either $g_1$ or $l_1$.
\end{itemize}

Finally, we consider the starting configuration $g_{init} ( n_{init}, l_{init} )$ where $g_{init}$ and $l_{init}$ are respectively the initial valuations of the thread-local and procedure-local variables of the main thread, and $n_{init}$ its starting node.

\section{The reachability problem}

As described previously in section \ref{sdpn_program_model}, we can model the behaviour of a real multi-threaded program with a SDPN. Many static analysis techniques rely on being able to determine whether a given critical state is reachable or not from the starting configuration of a program.

Since checking reachability in a real program amounts to checking reachability in its corresponding SDPN w.r.t to the strict semantics, we want to solve the following \emph{reachability problem}: given a SDPN $M$ and two sets of configuration $C$ and $C'$, is there a configuration in $C'$ that is reachable from $C$ with regards to the \emph{strict} semantics?

It has unfortunately been proven by Ramalingam in \cite{R-acm00} that, even if $C$ and $C'$ are regular, this problem is undecidable for synchronization sensitive pushdown systems, hence, SDPNs. Therefore, we reduce this problem to an execution path analysis of SDPNs with \emph{relaxed} semantics.

\subsection{From the strict to the relaxed semantics}

It is easy to see that the following theorem holds:
\begin{theorem}
	\label{thm_strict_relaxed}
	Let $M$ be a SDPN and $c$, $c'$ two configurations of $M$; $c \dashrightarrow^*_{M} c'$ if and only if $\exists n \geq 0$ such that $c \Xrightarrow[\tau^n]_{M} c'$.
\end{theorem}
Intuitively, an execution path with regards to the relaxed semantics of the form $\tau^n$ only uses internal actions or synchronized actions between two threads: a synchronization signal $a$ is always paired with its co-action $\overline{a}$. Any configuration reachable using this path can be reached with regards to the strict semantics as well. Such a path is said to be \emph{perfectly synchronized}.

Therefore, the reachability problem amounts to determining whether:
\[Paths_M ( C, C' ) \cap \tau^* = \emptyset\]
that is, if there is an execution path from $C$ to $C'$ with regards to the relaxed semantics of the form $\tau^n$. Obviously, we can't always compute $Paths_M ( C, C' )$. Our idea is therefore to compute an abstraction (over-approximation) of the set $Paths_M ( C, C' )$ and check the emptiness of its intersection with $\tau^*$: if it is indeed empty, then $C'$ can't be reached from $C$ with regards to the strict semantics.

It is worth noting that a configuration $p'_1 w'_1 p'_2 w'_2$ reachable from $p_1 w_1 p_2 w_2$ with regards to the strict semantics by synchronizing two rules $p_1 w_1 \xrightarrow{a} p'_1 w'_1$ and $p_2 w_2 \xrightarrow{\overline{a}} p'_2 w'_2$ using the synchronization rule \textbf{(1)} can obviously be reached with regards to the relaxed semantics by applying these two rules sequentially, using rule \textbf{(3)} twice, although the resulting path would obviously not be perfectly synchronized. Hence, the following theorem holds:
\begin{theorem}
	\label{thm_sdpn_dpn}
	Let $M$ be a SDPN and $c$, $c'$ two configurations of $M$; $c'$ is reachable from $c$ w. r. t. the relaxed SDPN semantics if and only if it is reachable w. r. t. the DPN semantics.
\end{theorem}
It implies that, since we can compute $pre^* ( M, C )$ with regards to the DPN semantics thanks to Theorem \ref{thm_pre_DPN_regular}, we can compute it with regards to the relaxed SDPN semantics as well.

\subsection{Representing infinite sets of configurations}

In order to compute an abstraction of $Paths_M ( C, C' )$ we need to be able to finitely represent infinite sets of configurations of a SDPN $M$. To do so, we introduce a class of finite automata called $M$-automata:

\begin{definition}[Bouajjani et al. \cite{BMT-concur05}]
	Given a SDPN $M = ( Act, P, \linebreak \Gamma, \Delta )$, an $M$-automaton is a finite automaton $A = ( \Sigma, S, \delta, s_{init}, F)$ such that:
	\begin{itemize}
		\item $\Sigma = P \cup \Gamma$ is the input alphabet;
		\item the set of states $S = S_C \cup S_S$ can be partitioned in two disjoint sets $S_C$ and $S_S$;
		\item $\delta \subseteq S \times \Sigma \times S$ is the set of transitions;
		\item $\forall s \in S_C$ and $\forall p \in P$, there is at most a single state $s_p$ such that $( s, p, s_p ) \in \delta$; moreover, $s_p \in S_S$ and $s$ is the only predecessor of $s_p$; transitions from states in $S_C$ are always labelled with state symbols in $P$ and go to dedicated states in $S_S$;
		\item states in $S_C$ do not have exiting transitions labelled with letters in $\Gamma$;
		\item states in $S_S$ do not have exiting transitions labelled in $P$; transitions labelled with letters in $\Gamma$ always go to states in $S_S$;
		\item transitions from $S_S$ to $S_C$ are always labelled with $\varepsilon$; these are the only allowed $\varepsilon$-transitions in the $M$-automaton;
		\item $s_{init} \in S_C$ is the initial state;
		\item $F \subseteq S_C$ is the set of final states.
	\end{itemize}
\end{definition}

An $M$-automaton is designed in such a manner that every path accepting a configuration $p_1 w_1 \ldots p_n w_n$ is a sequence of sub-paths $s_i \xrightarrow{p_i}_\delta s_p \Xrightarrow[w_i]_\delta q \xrightarrow{\varepsilon}_\delta s_{i + 1}$ where $s_i \in S_C$, $s_{i+1} \in S_C$ and every state in the path $s_p \Xrightarrow[w_i]_\delta q$ is in $S_S$. Being a finite state automaton, an $M$-automaton accepts a regular language that is a subset of $Conf_M$. Any regular language in $( P \Gamma^* )^*$ can be accepted by an $M$-automaton, as shown in Figure \ref{fig:4-exMconf}.

\begin{figure}
	\centering
	\begin{tikzpicture}[shorten >=1pt, node distance=2.5cm, on grid, auto] 
	\node[state,initial] (s) {\textcolor{purple}{$s$}}; 
	\node[state] (sp1) [right=of s] {\textcolor{teal}{$s_{p_1}$}};
	\node[state] (q1) [right=of sp1] {\textcolor{olive}{$q_1$}};
	\node[state] (ss) [below=of s] {\textcolor{purple}{$s'$}}; 
	\node[state] (ssp2) [right=of ss] {\textcolor{teal}{$s'_{p_2}$}}; 
	\node[state] (q2) [right=of ssp2] {\textcolor{olive}{$q_2$}};
	\node[state] (qf) [right=of q2] {\textcolor{olive}{$q_F$}};
	\path[->] 
	(s) edge node {$p_1$} (sp1)
	(sp1) edge node {$\gamma_1$} (q1)
	(q1) edge[loop above] node {$\gamma_1$} (q1)
	(q1) edge[in=90, out=270] node {$\varepsilon$} (ss)
	(ss) edge node {$p_2$} (ssp2)
	(ssp2) edge node {$\gamma_2$} (q2)
	(q2) edge node {$\gamma_3$} (qf);
	\end{tikzpicture}
	\captionof{figure}{Accepting a regular set $p_1 \gamma_1^{+} p_2 \gamma_2 \gamma_3$ with an $M$-automaton.}
	\label{fig:4-exMconf}
\end{figure}
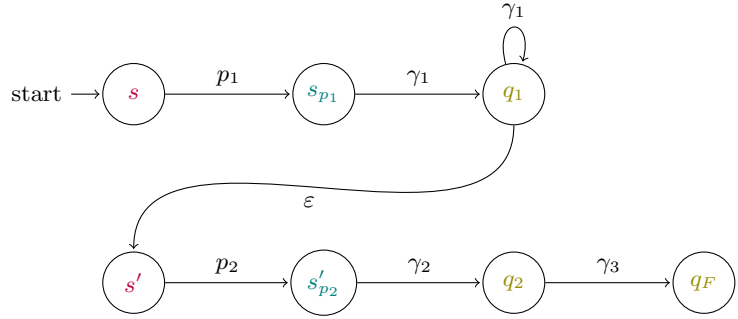

$M$-automata were introduced so that one could compute the set of predecessors of a DPN, hence, of a SDPN as well, by applying a saturation procedure to an $M$-automaton accepting the set of starting configurations, as shown by Bouajjani et al. in \cite{BMT-concur05}.

\section{Representing the set of paths}

In this section, we introduce an automata-theoretic representation of sets of synchronized paths by adding extra labels to $M$-automata.

\subsection{$\Pi$-configurations}

Let $\Pi = 2^{Act^*}$ be the set of all possible languages on $Act$. We define a $\Pi$-configuration of $M$ as a pair $(c, \pi) \in Conf_M \times Act^*$. We can extend the transition relation $\longrightarrow_{M}$ to $\Pi$-configurations with the following semantics: $\forall a \in Act$, if $c \xrightarrow{a}_M c'$, then $\forall \pi \in Act^*$, $( c, a \cdot \pi ) \longrightarrow_{M, \Pi} ( c', \pi )$. The configuration $( c, a \cdot \pi )$ is said to be an immediate $\Pi$-predecessor of $( c', \pi )$. The reachability relation $\leadsto_{M, \Pi}$ is the reflexive transitive closure of the relation $\longrightarrow_{M, \Pi}$.

Given a set of configurations $C$, we introduce the set of $\Pi$-predecessors $pre_{\Pi}^{*} ( M, C )$ of all $\Pi$-configurations $( c', \pi ) \in Conf_M \times Act^*$ such that $( c', \pi ) \linebreak \leadsto_{M, \Pi} ( c, \varepsilon )$ for $c \in C$. Obviously, we have:
\[ pre_{\Pi}^{*} ( M, C ) = \left\{ ( c', \pi ) \mid c' \in pre^* ( M, C ), \pi \in Paths_M ( \{ c' \}, C ) \right\}\]

Intuitively, $(c', \pi)$ is in $pre_{\Pi}^{*} ( M, C )$ if one can reach a configuration $c \in C$ from $c'$ by following a path $\pi$.

\subsection{The shuffle product}

Assuming we know the path languages of two different threads, we want to compute the path language of these two threads running in parallel.

Intuitively, this new language will be an interleaving of the two aforementioned sets, but can feature synchronized actions between the two threads as well.

To this end, we define inductively a shuffle operation $\shuffle : Act^* \times Act^* \rightarrow \Pi$ such that, given two paths, their shuffle product is the set of all possible interleaving (with synchronization) of these paths.

Let $w = a_1 \ldots a_n$ and $w' = b_1 \ldots b_m$ be two such paths:
\begin{itemize}
	\item $w \shuffle \varepsilon = \varepsilon \shuffle w = \{ w \}$;
	
	\item if $b_1 \neq \overline{a_1}$, then:
	$w \shuffle w' = a_1 \cdot [(a_2 \ldots a_n) \shuffle ( b_1 \ldots b_m )] \cup b_1 \cdot [(a_1 \ldots a_n) \shuffle (b_2 \ldots b_m)]$;
	
	\item if $b_1 = \overline{a_1}$, then:
	$w \shuffle w' = a_1 \cdot [(a_2 \ldots a_n) \shuffle ( b_1 \ldots b_m )] \cup b_1 \cdot [(a_1 \ldots a_n) \shuffle (b_2 \ldots b_m)] \cup \tau \cdot [(a_2 \ldots a_n) \shuffle (b_2 \ldots b_m)]$; two synchronized actions $a_1$ and $\overline{a_1}$ result in an internal action $\tau$, hence, there is a component $\tau \cdot ( w_{1} \shuffle w_{2} )$ of the shuffle product where the two paths synchronize.
\end{itemize}

The shuffle operation is obviously \emph{commutative} and \emph{associative}, as outlined in \cite{lothaire-1997}. We can extend naturally the operation $\shuffle$ to sets of paths: $\kappa_1 \shuffle \kappa 2 = \mathop{\bigcup} \limits_{\pi_1 \in \kappa_1, \pi_2 \in \kappa_2} (\pi_1 \shuffle \pi_2)$. It is still commutative and associative.

\subsection{$\Pi$-automata}

We represent sets of $\Pi$-configurations of a SDPN $M$ with a class of labelled $M$-automata, called $\Pi$-automata.

\begin{definition}
	Let $M = ( Act, P, \Gamma, \Delta )$ be a SDPN, a $\Pi$-automaton is a finite automaton $A = ( \Sigma, S, \delta, s_{init}, F )$ where $\Sigma = P \cup \Gamma$ is the input alphabet, $S = S_C \cup S_S$ is a finite set of control states with $S_C \cap S_S = \emptyset$, $\delta \subseteq ( S_C \times P \times S_S ) \cup ( S_S \times \Gamma \times \Pi^{\Pi} \times S_S ) \cup ( S_S \times \{ \varepsilon \} \times S_C )$ a finite set of transition rules (where $\Pi^\Pi$ is the set of functions from $\Pi$ to $\Pi$), $s_{init}$ an initial state, and $F$ a set of final states.
	
	Moreover, $A$ is such that, if we consider the projection $\delta_\Sigma$ of $\delta$ on $S \times \Sigma^* \times S$, ignoring labels in $\Pi^\Pi$, then $( \Sigma, S, \delta_\Sigma, s_{init}, F )$ is a $M$-automaton.
\end{definition}

Intuitively, a $\Pi$-automaton can be seen as an $M$-automaton whose transitions labelled by stack symbols in $\Gamma$ have been given an additional label in $\Pi^\Pi$. We can consider a simple $M$-automaton as a $\Pi$-automaton if we label each transition in $S_S \times \Gamma \times S_S$ with the identity function.

While it would be simpler to label transitions of a $M$-automaton with elements of $\Pi$, this representation would be flawed for the purpose of the algorithms outlined in Section \ref{labelling_constraints}. The intuition behind the use of functions in $\Pi^\Pi$ as labels is detailed there.

\subsubsection{The transition relation.}

Let $A$ be a $\Pi$-automaton. We define a simple transition relation $\longrightarrow_A$ according to the following semantics:
\begin{itemize}
	\item if $( s, p, s' ) \in \delta \cap ( S \times ( P \cup \{ \varepsilon \} ) \times S )$, then $s \xrightarrow{p}_A s'$;
	
	\item if $( s, \gamma, e, s' ) \in \delta \cap ( S_S \times \Gamma \times \Pi^{\Pi} \times S_S )$, then $s \xrightarrow{( \gamma, e )}_A s'$;
	
	\item if $s \xrightarrow{( w_{1}, e_{1} )}_A s_1$ and $s_1\xrightarrow{( w_{2}, e_{2} )}_A s'$, then $s \xrightarrow{( w_{1} w_2, e_{1} \circ e_{2} )}_A s'$, where $\circ$ is the composition operation on functions.
\end{itemize}

We then extend inductively this transition relation to a full path relation $\Longrightarrow_A \subseteq S \times \Sigma^* \times \Pi^{\Pi} \times S$:
\begin{itemize}
	\item for each $s \in S_S$, $s \xRightarrow{( \varepsilon, Id )}_A s$, where $Id$ stands for the identity function;
	
	\item if there is a sequence $s_{0} \xrightarrow{( \gamma_{1}, e_{1} )}_A s_{1} \ldots s_{n-1} \xrightarrow{( \gamma_{n}, e_{n} )}_A s_{n}$ with $s_0, \ldots, s_n \in S_S$, then $s_{0} \xRightarrow{( w, e )}_A s_{n}$, where $w = \gamma_{1} \ldots \gamma_{n}$ and $e = e_{1} \circ \ldots \circ e_{n}$; this is a simple sequence of actions along a single thread;
	
	\item if there is a sequence $s \xrightarrow{p_1}_A s_{p_1} \xRightarrow{( w_1, e_1 )}_A q \xrightarrow{\varepsilon}_A s' \xrightarrow{p_2}_A s'_{p_2} \xRightarrow{( w_2, e_2 )}_A q'$ such that $q', q \in S_S$ and $s, s' \in S_C$, then $s \xRightarrow{( w, e )}_A q$, where $w = p_1 w_1 p_2 w_2$ and $e: y \longrightarrow e_1 ( \{ \varepsilon \} ) \shuffle e_2 ( y )$; the automaton represents two parallel processes $p_1 w_1$ and $p_2 w_2$ whose abstract execution paths must be shuffled; moreover, since the first process will no longer be extended by further transitions of the automaton, we get rid of the variable of $e_1$ by considering $e_1 ( \{ \varepsilon \} )$ instead.
\end{itemize}
Note that this path relation is well-defined because $\shuffle$ is associative.

A path $s_{0} \xRightarrow{( c, e )}_A s_{n}$ is said to be an \emph{execution} of $A$ if $s_0 = s_{init}$. It is then said to be \emph{accepting} if $s_n \in F$. We then say that $A$ \emph{accepts} $( c, \pi )$ for all $\pi \in \Pi$ such that $\pi \in e ( \{ \varepsilon \} )$. This way, accepting execution paths in $\Pi$-automata can be used to represent whole sets of paths. We define the set $L_\Pi ( A )$ of all $\Pi$-configurations of $M$ accepted by $A$.

\section{Characterizing the set of paths}

Let $C$ be a regular set of configurations of a SDPN $M = ( Act, P, \Gamma, \Delta )$. We want to define a $\Pi$-automaton $A_{pre_{\Pi}^{*}}$ accepting $pre_{\Pi}^{*} ( M, C )$. Our intuition is to add extra labels in $\Pi^\Pi$ to the $M$-automaton accepting $pre^{*} ( M, C )$.

\subsection{Computing $pre^* ( M, C )$}
\label{compute_preDPN}

Given a SDPN $M$ and a regular set $C$ of configurations of $M$ accepted by an $M$-automaton $A$, we want to compute an $M$-automaton $A_{pre^*}$ accepting $pre^* ( M, C )$. Thanks to Theorem \ref{thm_sdpn_dpn}, we can apply the saturation procedure defined in \cite{BMT-concur05} to $A$. Let us remind this procedure. Initially, $A_{pre^*} = A$, then we apply the following rules until saturation to $A_{pre^*}$:
\begin{description}
	\item[$( R_1 )$] if $p \gamma \xrightarrow{a} p' w \in \Delta$ and $s \Xrightarrow[p'w]_{A_{pre^*}} s'$ for $s \in S_S$, $s' \in S$, then add $s_p \Xrightarrow[\gamma]_{A_{pre^*}} s'$;
	
	\item[$( R_2 )$] if $p \gamma \xrightarrow{a} p_1 \gamma_1 p_2 \gamma_2 \in \Delta$ and $s \xrightarrow{p_1 \gamma_1 p_2 \gamma_2}\!\!^*_{A_{pre^*}} s'$ for $s \in S_S$, $s' \in S$, then add $s_p \Xrightarrow[\gamma]_{A_{pre^*}} s'$.
\end{description}
$\longrightarrow^*_{A_{pre^*}}$ stands for the transitive closure of the transition relation on the finite state automaton $A_{pre^*}$. The initial and final states remain the same.

Let us remind the intuition of these rules. We consider a sub-path $s \Xrightarrow[p' w]_{A_{pre^*}} s'$. By design of an $M$-automaton, $s$ should be in $S_C$ and there should be a path $s \xrightarrow{p'} s_{p'} \Xrightarrow[w] s'$ in the automaton. If we apply the saturation rule $( R_1 )$, we add an edge $s_p \xrightarrow{p} s'$ to $A_{pre^*}$ and create a sub-path $s \Xrightarrow[p \gamma] s'$ in the automaton. Therefore, if $A_{pre^*}$ accepts a configuration $u_1 p' w' u_2$ with a path $s_{init} \Xrightarrow[u_1] s_{p'} \Xrightarrow[p' w] s' \Xrightarrow[u_2] q_F$, $q_F \in F$, then it will accept its predecessor $u_1 p \gamma u_2$ as well with a path $s_{init} \Xrightarrow[u_1] s_p \Xrightarrow[p \gamma] s' \Xrightarrow[u_2] q_F$. The role of $( R_2 )$ is similar.

Thus, when this saturation procedure ends, the $M$-automaton $A_{pre^*}$ accepts the regular set $pre^* ( M, C )$.

\subsection{From $pre^* ( M, C )$ to $pre^*_\Pi ( M, C )$}
\label{labelling_constraints}

Given a SDPN $M$ and a regular set $C$ of configurations of $M$ accepted by an $M$-automaton $A$, we want to compute a $\Pi$-automaton $A_{pre^*_\Pi}$ accepting $pre^*_\Pi ( M, C )$. To this end, we will add new labels to the $M$-automaton $A_{pre^*}$. Our intuition is the following: $A_{pre_{\Pi}^{*}}$ should be such that if we have $( c', \pi ) \leadsto_{M, \Pi} ( c, \varepsilon )$, $c \in C$, then $c$ can be reached from $c'$ by a path $\pi$.

In order to compute $A_{pre_{\Pi}^{\pi}}$, we proceed as follows: we first accept configurations in $C$ with the path $\varepsilon$, then, from there, set constraints on the labelling functions of transitions of $A_{pre^*}$ depending on the relationship between edges introduced by the previous saturation procedure. This way, we build iteratively a set of constraints whose least solution is the set of execution paths from $pre^*( M, C )$ to $C$.

To this end, each transition $t$ in $A_{pre^*}$ labelled in $\Gamma$ is given a second label $\lambda ( t ) \in \Pi^\Pi$. To do so, we compute a set of constraints whose smallest solution (according to the order $\subseteq$ of the language lattice) will be the labels $\lambda ( t )$. If $t = q_1 \xrightarrow{\gamma} q_2$, then we write $\lambda ( t ) = \lambda ( q_1, \gamma, q_2 )$.

\subsubsection{The need for functions.}

We will explain intuitively here why we label the automaton with functions and not with sets of paths in $\Pi$.

Let us consider a $M$-automaton labelled by sets of paths as shown in Figure \ref{fig:func1}. To a thread in configuration $p_1 \gamma_1$, we match a set $\{ a \}$ and to a thread in $p_2 \gamma_2 \gamma_3$, we match $\{ b, c \}$. We assume there is a rule $p \gamma \xrightarrow{d} p_1 \gamma_1 p_2 \gamma_2$ in $M$. By applying a saturation rule of the algorithm outlined in the previous section, we add a new dotted transition ${s_1}_p \xrightarrow{\gamma} q_2$. Intuitively, we label it with $d \cdot (\{ a \} \shuffle \{ b \}) = \{ dab, dba \}$: the label of the spawn action, followed by the synchronization of the paths matched to the two resulting threads.

\begin{figure}
		\centering
		\begin{tikzpicture}[shorten >=1pt, node distance=2.5cm and 2.5cm, on grid, auto] 
		\node[state,initial] (s1) {$s_1$}; 
		\node[state] (s1p1) [right=of s1] {${s_1}_{p_1}$};
		\node[state] (q1) [right=4cm of s1p1] {$q_1$};
		\node[state] (s2) [right=of q1] {$s2$}; 
		\node[state] (s2p2) [below=of s2] {${s_2}_{p_2}$}; 				
		\node[state] (q2) [below=of s2p2] {$q_2$};
		\node[state] (q3) [below=of q2] {$q_3$};
		\node[state] (s1p) [below=of s1] {${s_1}_{p_1}$}; 

		\path[->] (s1) edge node {$p_1$} (s1p1)
		(s1p1) edge node {$( \gamma_1, \{ a \} )$} (q1)
		(q1) edge node {$\varepsilon$} (s2)
		(s2) edge node {$p_2$} (s2p2)
		(s2p2) edge node[xshift=-2cm] {$( \gamma_2, \{ b \} )$} (q2)
		(q2) edge node[xshift=-2cm] {$( \gamma_3, \{ c \} )$} (q3)
		(s1) edge node {$p$} (s1p)
		(s1p) edge[out=270, in=180, densely dotted] node[xshift=-2cm, yshift=-1cm] {$(\gamma, d \cdot (\{ a \} \shuffle \{ b \}) = \{ dab, dba \} )$} (q2);
		\end{tikzpicture}
		\captionof{figure}{Using labels in $\Pi$.}
		\label{fig:func1}
\end{figure}
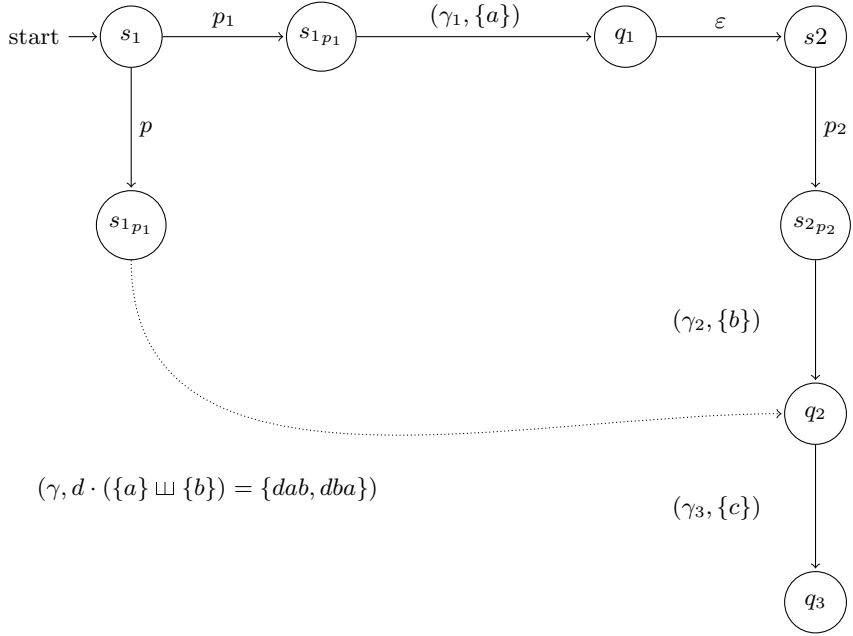

However, assuming semantics similar to $\Pi$-automata, the automaton in Figure \ref{fig:func1} would accept the configuration $(p \gamma \gamma_3, \{ dabc, dbac \})$ by going through states $s_1$, $p_1$, $q_2$, and $q_3$. But intuitively, we want to accept the set $\{ dabc, dbac, dbca \} = d \odot (\{ a \} \shuffle \{ bc \})$ instead, as the action $c$ can appear before $a$ in the interleaving of the execution paths matched to the two threads.

We can consider instead the $\Pi$-automaton shown in Figure \ref{fig:func2}. The new dotted transition is labelled by the function $x \rightarrow d \cdot [ a \shuffle (b \cdot x) ] = \{dabx, dbax, dbxa \}$. The variable $x$ stands for the end of the paths matched to the second thread (in this case, the action $c$) that are not examined by the saturation rule but would nonetheless have to be shuffled with the path $a$ of the first thread.

\begin{figure}
		\centering
		\begin{tikzpicture}[shorten >=1pt, node distance=2.5cm and 2.5cm, on grid, auto] 
		\node[state,initial] (s1) {$s_1$}; 
		\node[state] (s1p1) [right=of s1] {${s_1}_{p_1}$};
		\node[state] (q1) [right=4.5cm of s1p1] {$q_1$};
		\node[state] (s2) [right=of q1] {$s2$}; 
		\node[state] (s2p2) [below=of s2] {${s_2}_{p_2}$}; 				
		\node[state] (q2) [below=of s2p2] {$q_2$};
		\node[state] (q3) [below=of q2] {$q_3$};
		\node[state] (s1p) [below=of s1] {${s_1}_{p_1}$}; 

		\path[->] (s1) edge node {$p_1$} (s1p1)
		(s1p1) edge node {$( \gamma_1, x \rightarrow a \cdot x )$} (q1)
		(q1) edge node {$\varepsilon$} (s2)
		(s2) edge node {$p_2$} (s2p2)
		(s2p2) edge node[xshift=-3cm] {$( \gamma_2, x \rightarrow b \cdot x )$} (q2)
		(q2) edge node[xshift=-3cm] {$( \gamma_3, x \rightarrow c \cdot x )$} (q3)
		(s1) edge node {$p$} (s1p)
		(s1p) edge[out=270, in=180, densely dotted] node[xshift=-2cm, yshift=-1cm] {$(\gamma, x \rightarrow d \cdot [ a \shuffle (b \cdot x) ])$} (q2);
		\end{tikzpicture}
		\captionof{figure}{Using labels in $\Pi^\Pi$.}
		\label{fig:func2}
\end{figure}
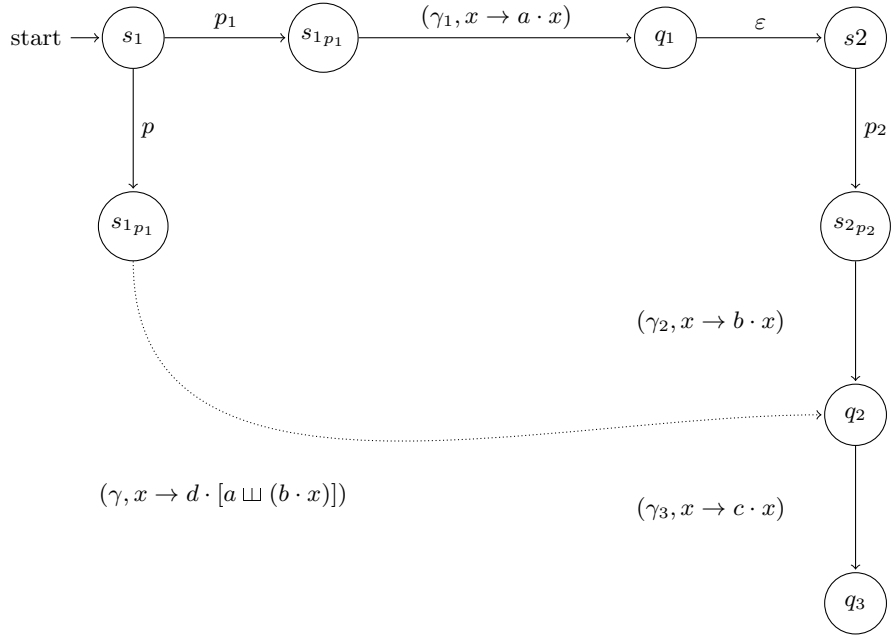

The automaton in Figure \ref{fig:func2} has an execution path labelled by $(p \gamma \gamma_3, f : x \rightarrow \{ dabcx, dbacx, dbcxa \})$ when it goes through states $s_1$, $p_1$, $q_2$, and $q_3$. Therefore, it accepts all the paths $f ( \{ \varepsilon \} ) = \{ dabc, dbac, dbca \}$.

\subsubsection{The constraints.}

For two functions in $\Pi^\Pi$, we write that $f \subseteq g$ if $\forall x \in \Pi$, $f ( x ) \subseteq g ( x )$. We now consider the following set of constraints on the labels of transitions of $A_{pre^*}$ in $S_S \times \Gamma \times S_S$, where $Q$ is the set of states of $A$:

\begin{description}
	\item[$( Z_1 )$] if $t$ belongs to $A$, then: \[ Id \subseteq \lambda ( t ) \]
	
	\item[$( Z_2 )$] for each rule $p \gamma \xrightarrow{a} p' \gamma' \in \Delta$, for each $q \in Q$, for each $s \in S_c$: \[ a \cdot \lambda ( s_{p'}, \gamma', q ) \subseteq \lambda ( s_{p}, \gamma, q ) \]
	
	\item[$( Z_3 )$] for each rule $p \gamma \xrightarrow{a} p' \varepsilon \in \Delta$, for each $s \in S_c$: \[a \cdot Id \subseteq \lambda ( s_{p}, \gamma, s_{p'} ) \]
	
	\item[$( Z_4 )$] for each rule $p \gamma \xrightarrow{a} p' \gamma_1 \gamma_2 \in \Delta$, for each $q \in Q$, for each $s \in S_c$: \[\mathop{\bigcup} \limits_{q' \in Q} a \cdot ( \lambda ( s_{p'}, \gamma_{1}, q' ) \circ \lambda ( q', \gamma_{2}, q ) ) \subseteq \lambda ( s_{p}, \gamma, q ) \] 
	
	\item[$( Z_5 )$] for each rule $p \gamma \xrightarrow{a} p_{2} \gamma_{2} p_{1} \gamma _{1} \in \Delta$, for each $q \in Q$, for each $s \in S_c$: \[\mathop{\bigcup} \limits_{s'' \xrightarrow{\varepsilon}_{A_{pre^*}} s'} a \cdot ( \lambda ( s_{p_2}, \gamma_{2}, s'' ) ( \{ \varepsilon \} ) \shuffle \lambda (s'_{p_1}, \gamma_{1} ,q ) ) \subseteq \lambda ( s_{p}, \gamma, q ) \]
\end{description}

In a manner similar to \cite{BET-infinity04}, we eventually define the labels of $A_{pre_{\Pi}^{*}}$ as the least solution of the set of constraints outlined above in the complete lattice of functions in $\Pi^\Pi$. By Tarski Theorem, this solution exists. The following theorem holds:
\begin{theorem}
	\label{thm_preK}
	Let $M$ be a SDPN and $A$ an $M$-automaton accepting a regular set of configurations $C$. Then the $\Pi$-automaton $A_{pre_{\Pi}^{*}}$ accepts the set $pre_{\Pi}^{*} ( M, C )$.
\end{theorem}

Note that it doesn't mean we can compute the labels: an iterative computation of the least solution may not terminate. We now explain intuitively the meaning of these constraints.

\subsubsection{The intuition.}

If $c$ is a configuration of $C$, then $A_{pre^*_\Pi}$ should accept $( c, \varepsilon )$. This is expressed by constraint $( Z_1 )$.

Let $c' = p' \gamma' w \in pre^* ( M, C )$. If $p \gamma \xrightarrow{a} p' \gamma' \in \Delta$ and $( c', \pi ) \in pre^*_\Pi ( M, C )$, then $c = p \gamma w \in pre^* ( M, C )$, $( c, a \cdot \pi ) \leadsto_{M, \Pi} ( c', \pi )$, and $ ( c, a \cdot \pi ) \in pre^*_\Pi ( M, C )$. Hence, if $A_{pre^*_\Pi}$ accepts $( c', \pi )$ and uses a transition $s_{p'} \xrightarrow{\gamma'} q$ while doing so, then it should accept $( c, a \cdot \pi )$ as well using a transition $s_{p} \xrightarrow{\gamma} q$, as shown in Figure \ref{fig:4-switch}. This is expressed by constraint $( Z_2 )$.

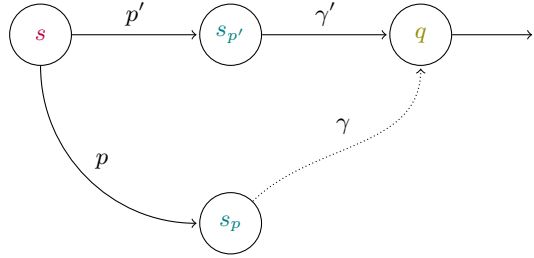
\begin{figure}
	\centering
	\begin{tikzpicture}[shorten >=1pt, node distance=2.5cm, on grid, auto] 
	\node[state] (s) {\textcolor{purple}{$s$}}; 
	\node[state] (spp) [right=of s] {\textcolor{teal}{$s_{p'}$}};
	\node[state] (q) [right=of spp] {\textcolor{olive}{$q$}};
	\node[state] (sp) [below=of spp] {\textcolor{teal}{$s_{p}$}}; 
	\path[->] 
	(s) edge node {$p'$} (spp)
	(spp) edge node {$\gamma'$} (q)
	(s) edge[out=270, in=180] node {$p$} (sp)
	(sp) edge[densely dotted, in=270] node {$\gamma$} (q)
	(q) edge node[above] {} +(1.5,0);
	\end{tikzpicture}
	\captionof{figure}{Case of a switch rule.}
	\label{fig:4-switch}
\end{figure}

Let $c' = p' w \in pre^* ( M, C )$. If $p \gamma \xrightarrow{a} p' \varepsilon \in \Delta$ and $( c', \pi ) \in pre^*_\Pi ( M, C )$, then $c = p \gamma w \in pre^* ( M, C )$, $( c, a \cdot \pi ) \leadsto_{M, \Pi} ( c', \pi )$, and $ ( c, a \cdot \pi ) \in pre^*_\Pi ( M, C )$. Hence, if $A_{pre^*_\Pi}$ accepts $( c', \pi )$, then it should accept $( c, a \cdot \pi )$ as well using a transition $s_{p} \xrightarrow{\gamma} s'_p$, as shown in Figure \ref{fig:4-pop}. This is expressed by constraint $( Z_3 )$.

\begin{figure}
	\centering
	\begin{tikzpicture}[shorten >=1pt, node distance=2.5cm, on grid, auto] 
	\node[state] (s) {\textcolor{purple}{$s$}}; 
	\node[state] (spp) [right=of s] {\textcolor{teal}{$s_{p'}$}};
	\node[state] (sp) [below=of spp] {\textcolor{teal}{$s_{p}$}};
	\path[->] 
	(s) edge node {$p'$} (spp)
	(s) edge[out=270, in=180] node {$p$} (sp)
	(sp) edge[densely dotted] node {$\gamma$} (spp)
	(spp) edge node[above] {} +(1.5,0);
	\end{tikzpicture}
	\captionof{figure}{Case of a pop rule.}
	\label{fig:4-pop}
\end{figure}
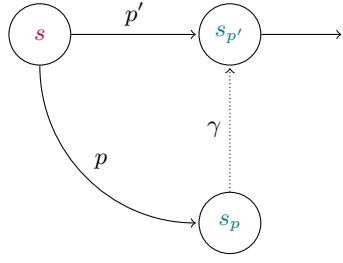

Let $c' = p' \gamma_1 \gamma_2 w \in pre^* ( M, C )$. If $p \gamma \xrightarrow{a} p' \gamma_1 \gamma_2 \in \Delta$ and also $( c', \pi ) \in pre^*_\Pi ( M, C )$, then $c = p \gamma w \in pre^* ( M, C )$, $( c, a \cdot \pi ) \leadsto_{M, \Pi} ( c', \pi )$, and $ ( c, a \cdot \pi ) \in pre^*_\Pi ( M, C )$. Hence, if $A_{pre^*_\Pi}$ accepts $( c', \pi )$ and uses two transition $s_{p'} \xrightarrow{\gamma_1} q'$ and $q' \xrightarrow{\gamma_2} q$ while doing so, then it should accept $( c, a \cdot \pi )$ as well using a transition $s_{p} \xrightarrow{\gamma} q$, as shown in Figure \ref{fig:4-push}. Moreover, there can be many possible intermediate states $q'$ between $s'_p$ and $q$ such that $s_{p'} \xrightarrow{\gamma_1} q'$ and $q' \xrightarrow{\gamma_2} q$. In the automaton $A_{pre^*_\Pi}$, the path $\pi$ should therefore be represented by the union for all possible intermediate state $q'$ of the concatenation of the two labelling functions $\lambda (s_{p'}, \gamma_1, q' )$ and $\lambda ( q', \gamma_2, q )$. This is expressed by constraint $( Z_4 )$.

\begin{figure}
	\centering
	\begin{tikzpicture}[shorten >=1pt, node distance=2.5cm, on grid, auto] 
	\node[state] (s) {\textcolor{purple}{$s$}}; 
	\node[state] (spp) [right=of s] {\textcolor{teal}{$s_{p'}$}};
	\node[state] (qq1) [right=of spp] {\textcolor{olive}{$q'_1$}};
	\node[state] (qq2) [below=of qq1] {\textcolor{olive}{$q'_2$}};
	\node[state] (q) [right=of qq2] {\textcolor{olive}{$q$}};
	\node[state] (sp) [below=of spp] {\textcolor{teal}{$s_{p}$}}; 
	\path[->] 
	(s) edge node {$p'$} (spp)
	(spp) edge node {$\gamma_1$} (qq1)
	(spp) edge[out=315, in=180] node {$\gamma_1$} (qq2)
	(qq1) edge node {$\gamma_2$} (q)
	(qq2) edge node {$\gamma_2$} (q)
	(s) edge[out=270, in=180] node {$p$} (sp)
	(sp) edge[densely dotted, out=315, in=270] node {$\gamma$} (q)
	(q) edge node[above] {} +(1.5,0);
	\end{tikzpicture}
	\captionof{figure}{Case of a push rule.}
	\label{fig:4-push}
\end{figure}
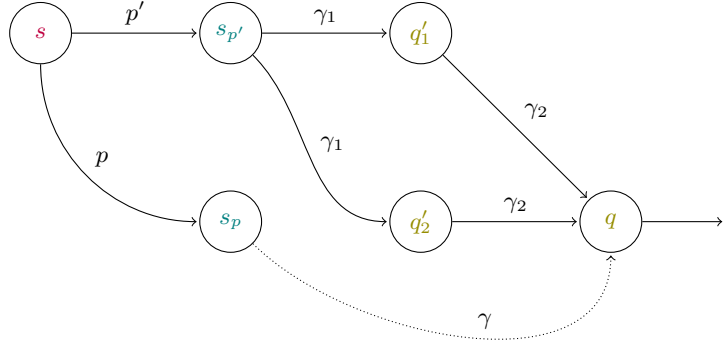

Let $c' = p_2 \gamma_2 p_1 \gamma_1 w \in pre^* ( M, C )$. If $p \gamma \xrightarrow{a} p_{2} \gamma_{2} p_{1} \gamma _{1} \in \Delta$ and $( c', \pi ) \in pre^*_\Pi ( M, C )$, then $c = p \gamma w \in pre^* ( M, C )$, $( c, a \cdot \pi ) \leadsto_{M, \Pi} ( c', \pi )$, and $ ( c, a \cdot \pi ) \in pre^*_\Pi ( M, C )$, as shown in Figure \ref{fig:4-synch}. The two processes $p_2 \gamma_2$ (thread 2 in \ref{fig:4-synch}) and $p_1 \gamma_1$ (thread 1 in \ref{fig:4-synch}) are interleaved, hence, their execution paths must be shuffled: if $\pi_1$ is an execution path associated to $p_1 \gamma_1$, and $\pi_2$, to $p_2 \gamma_2$, then an path $\pi' = \pi_2 \shuffle \pi_1$ should be associated to $p_2 \gamma_2 p_1 \gamma_1$. Moreover, if we consider a path $s_{p_2} \xrightarrow{\gamma_2} s'' \xrightarrow{\varepsilon} s' \xrightarrow{p_1} s'_{p_1} \xrightarrow{\gamma_1} q$ in the automaton $A_{pre^*_\Pi}$, then no path $\pi_2$ associated to $p_2 \gamma_2$ can be extended further, and should therefore be represented by $( \lambda ( s_{p_2}, \gamma_{2}, s' ) ( \{ \varepsilon \} ) )$. Again, we must also consider each possible intermediate state $s''$ in the previous path, hence, an union of functions. This is expressed by constraint $( Z_5 )$.

\begin{figure}
	\centering
	\begin{tikzpicture}[shorten >=1pt, node distance=2.5cm, on grid, auto] 
	\node[state] (s) {\textcolor{purple}{$s$}}; 
	\node[state, label={\textcolor{blue}{\small Thread 2}}] (sp2) [right=of s] {\textcolor{teal}{$s_{p_2}$}};
	\node[state] (q) [right=of sp2] {\textcolor{olive}{$q'$}};
	\node[state] (ss) [right=of q] {\textcolor{purple}{$s'$}};
	\node[state, label={\textcolor{green}{\small Thread 1}}] (ssp1) [right=of ss] {\textcolor{teal}{$s''_{p_1}$}};
	\node[state] (qq) [below=of ssp1] {\textcolor{olive}{$q$}};
	\node[state] (sp) [below=of spp] {\textcolor{teal}{$s_{p}$}};
	\path[->] 
	(s) edge node {\textcolor{blue}{$p_2$}} (sp2)
	(sp2) edge node {\textcolor{blue}{$\gamma_2$}} (q)
	(q) edge node {$\varepsilon$} (ss)
	(ss) edge node {\textcolor{green}{$p_1$}} (ssp1)
	(ssp1) edge node {\textcolor{green}{$\gamma_1$}} (qq)
	(s) edge[out=270, in=180] node {$p$} (sp)
	(sp) edge[densely dotted] node {$\gamma$} (qq)
	(qq) edge node[above] {} +(1.5,0);
	\end{tikzpicture}
	\captionof{figure}{Case of a spawn rule.}
	\label{fig:4-synch}
\end{figure}
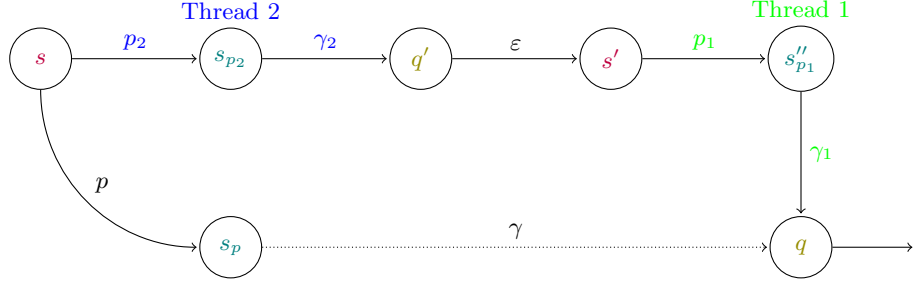

\subsection{Proof of Theorem \ref{thm_preK}}
\label{proof}

In order to prove Theorem \ref{thm_preK}, we first show that $A_{pre_{\Pi}^{*}}$ accepts every $\Pi$ predecessor of $C$.
\begin{lemma}
	\label{lemma_1}
	We consider $c = p_{1} v_{1} \ldots p_{n} v_{n} \in C$, $s_{init}$ the initial state of $A$, and $F$ its set of final states. If $( c'= ( p'_{1} w_{1} \ldots p'_{l} w_{l} ), \pi ) \leadsto_{M, \Pi} ( c, \varepsilon )$ for $\pi \in \Pi_{\Pi}$, then $\exists e \in \Pi^\Pi$ and $q_F \in F$ such that $\pi \le e ( \varepsilon )$ and $s_{init} \xRightarrow{( c', e )}_{A_{pre_\Pi^*}} q_{F}$.
\end{lemma}

We then prove that every configuration accepted by $A_{pre_{\Pi}^{*}}$ is a $\Pi$ predecessor of $C$.
\begin{lemma}
	\label{lemma_2}
	$\forall \pi \in Act^*$, if there is an accepting execution $s_{init} \xRightarrow{( c', e )}_{A_{pre_{\Pi}^{*}}} q_{F}$ such that $c' = ( p'_{1} w'_{1} \ldots p'_{n} w'_{n} )$ and $\pi \le e ( \{ \varepsilon \} )$, there is a configuration $c$ such that $s_{init} \overset{c}{\rightarrow}_A q_{F} \in F$ and $( c', \pi ) \leadsto_{M, \Pi} ( c, \varepsilon )$.
\end{lemma}

\subsubsection{Proof of Lemma \ref{lemma_1}}

We can expand the labelling of functions to paths by composing labels along the paths, in a manner similar to the semantics of $\Pi$-automata. Hence, $\lambda(s, w, q)$ is such that $s \xRightarrow{( w, \lambda(s, w, q) )}_A q$.

We prove by induction on $k$ that if $c'= ( ( p'_{1} w'_{1} \ldots p'_{l} w'_{l} ), \pi ) \rightarrow^k_{M, \Pi} ( c, \varepsilon )$, then $\exists e \in \Pi^\Pi$ and $q_F \in F$ such that $\pi \in e ( \{ \varepsilon \} )$ and $s_{init} \xRightarrow{( c, e )}_{A_{pre_\Pi^*}} q_{F}$.
\begin{description}
	\item[Basis:] if $k = 0$, then $c' = c$, and $\pi = \varepsilon$. Because of constraint $( Z_1 )$, each transition $t$ that exists both in $A$ and $A_{pre_\Pi^*}$ is such that $Id \subseteq \lambda ( t )$. Hence, if we follow an accepting execution $s_{init} \xRightarrow{( c, e )}_{A_{pre_\Pi^*}} q_{F}$ only using transitions in $A$, it must be such that $\varepsilon \in e( \{ \varepsilon \} )$. Hence, $A_{pre_{\Pi}^{*}}$ accepts $(c, \varepsilon)$.
	
	\item[Induction step:] we consider $( c_1 = ( p''_{1} u_{1} \ldots p''_{j} u_{j} ), \pi' )$ such that $( c', \pi ) \leadsto_{M, \Pi} ( c_1, \pi' )$ and $( c_1, \pi' ) \leadsto^k_{M, \Pi} ( c, \varepsilon )$. We have the two following cases:
	\begin{itemize}
		\item if $j = l + 1$, a new process has been created from $c'$ to $c_1$ by a rule of the form $r= p'_{i} \gamma \xrightarrow{a} p''_{i} u_{i} p''_{i+1} u'_{i+1}$. Then, there are $t_{1}$, $t_{2}$ in $(P\Gamma^*)^*$ such that $c_{1} = t_{1} p''_{i} u_{i} p''_{i+1} u_{i+1} t_{2}$ and $c' = t_{1} p'_{i} w'_{i} t_{2}$, and there is $u \in \Gamma^*$ such that $w'_{i} = \gamma u$ and $u_{i+1} =u'_{i+1} u$.
		
		By induction, as shown in figure \ref{fig:lm1a}, there is a path in $A_{pre_{\Pi}^{*}}$ such that it accepts $( c_1, \pi')$ and $\pi' \in \lambda_1 ( \{ \varepsilon \} ) \shuffle \lambda_2 ( \{ \varepsilon \} ) \shuffle \lambda_3 \circ \lambda_4 ( \{ \varepsilon \} ) \shuffle \lambda_5 ( \{ \varepsilon \} )$, where $\lambda_1 = \lambda( s_{init}, t_1, s )$, $\lambda_2 = \lambda( s_1, p''_i u', s_2 )$, $\lambda_3 = \lambda( s_3, p''_{i+1} u'_{i+1}, s_4 )$, $\lambda_4 = \lambda( s_4, u, s_5 )$ and $\lambda_5 = \lambda(s_6, t_2, q_F)$. The saturation procedure creates a transition $( {s_1}_{p'_{i}}, \gamma, s_{4} ) = a \cdot ( \lambda_2 ( \{ \varepsilon \} ) \shuffle \lambda_3 )$. Hence:
		
		\begin{figure}
			\centering
			\begin{tikzpicture}[shorten >=1pt, node distance=2.5cm and 4.25cm, on grid, auto] 
			\node[state,initial] (si) {$s_{init}$}; 
			\node[state] (s) [above=of si] {$s$};
			\node[state] (s1) [right=of s] {$s_1$};
			\node[state] (s1pi) [below=of s1] {${s_1}_{p'_i}$}; 
			\node[state] (s2) [right=of s1] {$s_2$}; 				
			\node[state] (s4) [below=of s1pi] {$s_4$};
			\node[state] (s3) [right=of s4] {$s_3$}; 
			\node[state] (s5) [left=of s4] {$s_5$};
			\node[state] (s6) [below=of s5] {$s_6$};
			\node[state] (qf) [right=of s6] {$q_F$};
			\path[->] 
			(si) edge node {$( t_1, \lambda_1 )$} (s)
			(s) edge node {$\varepsilon$} (s1)
			(s1) edge node {$( p''_i u_i, \lambda_2 )$} (s2)
			(s2) edge[out=270, in=90] node[xshift=-20pt] {$\varepsilon$} (s3)
			(s3) edge node {$( p''_{i+1} u'_{i+1}, \lambda_3 )$} (s4)
			(s4) edge node {$( u, \lambda_4 )$} (s5)
			(s5) edge[out=270, in=90] node {$\varepsilon$} (s6)
			(s6) edge node {$( t_2, \lambda_5 )$} (qf)
			(s1) edge[densely dotted] node {$p'_i$} (s1pi)
			(s1pi) edge[densely dotted] node {$\gamma$} (s4);
			\end{tikzpicture}
			\captionof{figure}{Adding an edge if a new process is spawned.}
			\label{fig:lm1a}
		\end{figure}
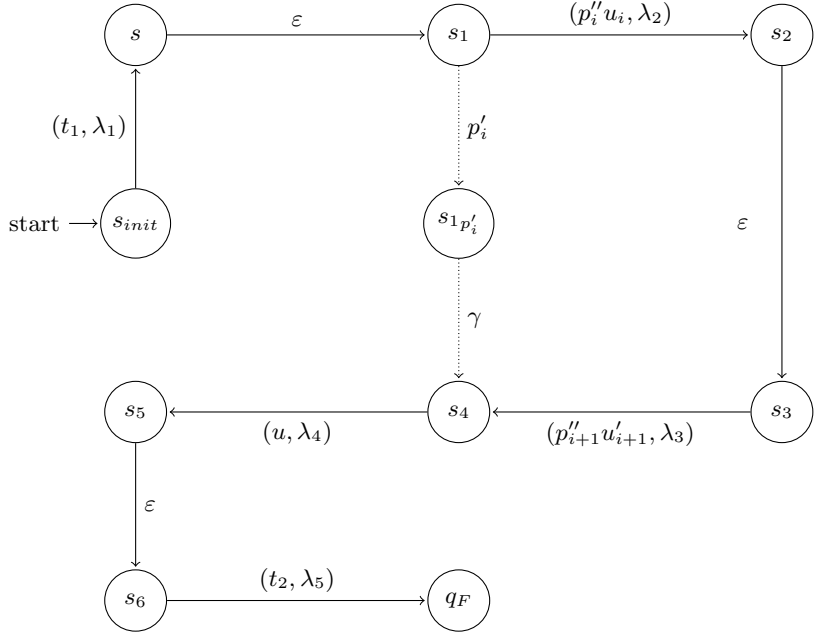	
		
		\[\pi' \in \lambda_1 ( \{ \varepsilon \} ) \shuffle \lambda_2 ( \{ \varepsilon \} ) \shuffle \lambda_3 \circ \lambda_4 ( \{ \varepsilon \} ) \shuffle \lambda_5 ( \{ \varepsilon \} )
		\]
		
		\[\pi = a \cdot \pi' \in a\cdot( \lambda_1 ( \{ \varepsilon \} ) \shuffle \lambda_2 ( \{ \varepsilon \} ) \shuffle \lambda_3 \circ \lambda_4 ( \{ \varepsilon \} ) \shuffle \lambda_5 ( \{ \varepsilon \} ) )
		\]
		
		\[\pi \in \lambda_1 ( \{ \varepsilon \} ) \shuffle ( a \cdot ( \lambda_2 ( \{ \varepsilon \} ) \shuffle \lambda_3 ) \circ \lambda_4 ( \{ \varepsilon \} ) ) \shuffle \lambda_5 ( \{ \varepsilon \} )
		\]
		
		Moreover, $A_{pre^*}$ accepts $c'$ with the path $s_{init} \xrightarrow{ t_{1} p'_{i} \gamma u t_{2} }_{A_{pre^*}} q_{F}$. But we have $\lambda( s_{init}, t_1 p'_i \gamma u t_2, q_F) \supseteq \lambda_1 \shuffle ( a \cdot ( \lambda_2 ( \{ \varepsilon \} ) \shuffle \lambda_3 ) \circ \lambda_4 ) \shuffle \lambda_5$ because of constraint $(Z_5)$. $A_{pre_{\Pi}^*}$ therefore accepts $\pi$.
		
		\item if $j = l$, no new process has been created while moving from the configuration $c'$ to the configuration $c_1$. Let $p'_{i} \gamma \xrightarrow{a}_M p''_{i} u'$ be the transition used to move from $c'$ to $c_1$.
		
		Let $t_1$ and $t_2$ be two words in $( P \Gamma^* )^*$ such that $c_{1} = ( t_{1} p''_{i} u_{i} t_{2} )$ and $c' = ( t_{1} p'_{i} w'_{i} t_{2} )$. There is $u \in \Gamma^*$ such that $w'_{i} = \gamma u$ and $u_{i} = u' u$.
		
		By induction, as shown in figure \ref{fig:lm1b}, there is a path $\pi'$ such that $A_{pre^*_\Pi}$ accepts $( c_1, \pi')$ and $\pi' \in \lambda_1 ( \{ \varepsilon \} ) \shuffle \lambda_2 \circ \lambda_3 ( \{ \varepsilon \} ) \shuffle \lambda_4 ( \{ \varepsilon \} )$, where $\lambda_1 \linebreak = \lambda ( s_{init}, t_1, s)$, $\lambda_2 = \lambda ( s_1, p''_i u', s_2 )$, $\lambda_3 = \lambda ( s_2, u, s_3)$ and $\lambda_4 = \lambda ( s_4, t_2, \linebreak q_F )$.
		
		\begin{figure}
			\centering
			\begin{tikzpicture}[shorten >=1pt, node distance=2.5cm, on grid, auto] 
			\node[state,initial] (si) {$s_{init}$}; 
			\node[state] (s) [right=of si] {$s$};
			\node[state] (s1) [right=of s] {$s_1$};
			\node[state] (s1pi) [right=of s1] {${s_1}_{p'_i}$};
			\node[state] (s2) [below=of s1] {$s_2$}; 
			\node[state] (s3) [left=of s2] {$s_3$}; 
			\node[state] (s4) [below=of s3] {$s_4$};
			\node[state] (qf) [right=of s4] {$q_F$};
			\path[->] 
			(si) edge node {$t_1$} (s)
			(s) edge node {$\varepsilon$} (s1)
			(s1) edge node {$p''_i u'$} (s2)
			(s2) edge node {$u$} (s3)
			(s3) edge[out=270, in=90] node {$\varepsilon$} (s4)
			(s4) edge node {$t_2$} (qf)
			(s1) edge[densely dotted] node {$p'_i$} (s1pi)
			(s1pi) edge[densely dotted] node {$\gamma$} (s2);
			\end{tikzpicture}
			\captionof{figure}{Adding an edge if no new process is spawned.}
			\label{fig:lm1b}
		\end{figure}
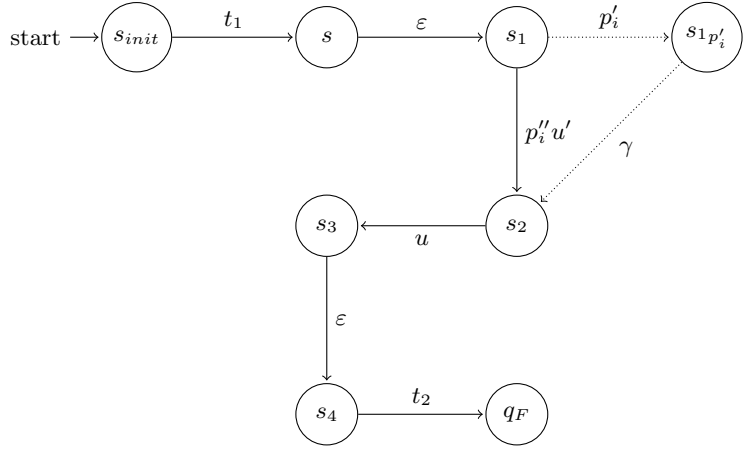
		
		The saturation procedures creates a transition $( {s_{1}}_{p'_{i}}, \gamma, s_{2} )$ such that we have $\lambda ( {s_{1}}_{p'_{i}}, \gamma, s_2 ) = a \cdot ( \lambda_2 )$. Moreover:
		
		\[\pi' \in \lambda_1( \{ \varepsilon \} ) \shuffle \lambda_2 \circ \lambda_3 ( \{ \varepsilon \} ) \shuffle \lambda_4 ( \{ \varepsilon \} )
		\]
		
		\[\pi = a \cdot \pi' \in a \cdot ( \lambda_1 ( \{ \varepsilon \}) \shuffle \lambda_2 \circ \lambda_3 ( \{ \varepsilon \} ) \shuffle \lambda_4 ( \{ \varepsilon \} ) )
		\]
		
		\[\pi \in \lambda_1 ( \{ \varepsilon \} ) \shuffle a \cdot \lambda_2 \circ \lambda_3 ( \{ \varepsilon \} ) \shuffle \lambda_4 ( \{ \varepsilon \} )
		\]
		
		The automaton $A_{pre^*}$ accepts $c'$ with a path $s_{init} \xrightarrow{ t_{1} p'_{i} \gamma u t_{2} }_{A_{pre^*}} q_{F}$ such that $\lambda( s_{init}, t_1 p'_i \gamma u t_2, q_F ) \supseteq \lambda_1 \shuffle a \cdot ( \lambda_2 \circ \lambda_3 ) \shuffle \lambda_4$ because of constraints $(Z_2, Z_3, Z_4)$. The automaton $A_{pre_{\Pi}^*}$ therefore accepts $\pi$.
	\end{itemize}
\end{description}

\subsubsection{Proof of Lemma \ref{lemma_2}}

We prove this lemma by induction the length $\left| \pi \right|$ of $\pi$.

\begin{description}
	\item[Basis:] if $\left| \pi \right| = 0$, then $\pi = \varepsilon$; because of constraints $(Z_2, Z_3, Z_4, Z_5)$, paths that use new transitions introduced by the saturation procedure are of length at least $1$ and $(c', \pi)$ can only be accepted by a sequence of transitions following constraint $( Z_1 )$, hence, transitions in $A$. Therefore, $c' \in C$ and the property holds.
	
	\item[Induction step:] if $\left| \pi \right| = k > 0$, let $\pi = {a_{1}} \cdot {a_{2}} \cdot \ldots \cdot {a_{k}}$. If $A_{pre_\Pi^*}$ accepts $(c', \pi)$, then there is an accepting execution $s_{init} = s_1 \xRightarrow{( p'_1 w'_1, e_1 )}_{A_{pre_\Pi^*}} s_2 \ldots s_n \linebreak \xRightarrow{( p'_n w'_n, e_1 )}_{A_{pre_\Pi^*}} s_{n+1} \in F$ such that $\pi \in e_1 ( \{ \varepsilon \} ) \shuffle e_2 ( \{ \varepsilon \} )\ldots \shuffle e_n ( \{ \varepsilon \} )$, with $e_i = \lambda ( s_i, p'_i w'_i, s_{i+1} )$.
	
	We have the two following cases:
	\begin{itemize}
		\item if $a_1 \neq \tau$, for each $i \in \{ 1, \ldots, n \}$, we split $e_i ( \{ \varepsilon \} )$ into two parts: ${a_1} \cdot S_i$, which represents the path expressions in $e_i$ starting with ${a_1}$, and $S'_i$, which stands for the path expressions starting with another symbol than ${a_1}$. We have $e_j ( \{ \varepsilon \} ) = ( {a_1} \cdot S_j ) \cup S'_j$ and $\pi = a_1 a_2 \ldots a_k \in ( ( {a_1} \cdot S_1 ) \cup S'_1 ) \shuffle \ldots ( ( {a_1} \cdot S_n ) \cup S'_n )$.
		
		Let $i$ be such that ${a_2} {a_3} \ldots {a_k} \in ( ( {a_1} \cdot S_1 ) \cup S'_1 ) \shuffle \ldots \shuffle ( S_i ) \shuffle \ldots ( ( {a_1} \cdot S_n ) \cup S'_n )$. The first symbol ${a_1}$ of $\pi$ must appear in one (let's say the $i$-th) of the $n$ sets $e_1 ( \{ \varepsilon \} ), \ldots, e_n ( \{ \varepsilon \} )$.
		
		We define $w'_i = \gamma_1 \gamma_2 \ldots \gamma_j$. We now consider the sub-sequence $s_i = s_{i_1} \linebreak \xRightarrow{( p'_i \gamma_1, \lambda_1)}_{A_{pre_\Pi^*}} s_{i_2} \ldots \xRightarrow{( \gamma_j, \lambda_j)} s_{i_{j+1}}$ of the accepting execution outlined earlier. We have $e_i = \lambda_1 \circ \lambda_2 \ldots \circ \lambda_j$.
		
		Since there are paths starting with ${a_1}$ in $e_i$, then there are words starting with ${a_1}$ in $\lambda_1$ as well. Therefore, there is a rule $r = p'_{i} \gamma_{1} \xrightarrow{a_1}_M q_{i} u$ in $M$ and a transition $( {s_i}_{q_{i}}, u, s_{i_2} )$ in $A_{pre^*}$, from which a transition $( {s_i}_{p'_{i}}, \gamma_{1}, s_{i_2} )$ labelled by $\lambda_1$ can be added by the saturation rules in such a manner that the inequality ${a_1} \cdot \lambda ( {s_i}_{q_i}, u, s_{i_2}) \subseteq \lambda_1$ holds.
		
		The automaton $A_{pre_{\Pi}^{*}}$ then has an accepting execution labelled by: $$c_{1} = ( p'_{1} w'_{1} \ldots p'_{i-1} w'_{i-1} q_{i} u \gamma_{2} \ldots \gamma_{j} p'_{i+1} w'_{i+1} \ldots p'_{n} w'_{n} )$$ and: $$( e_1 \shuffle \ldots e_{i-1} \shuffle ( \lambda' \circ \lambda_2 \ldots \circ \lambda_j ) \shuffle e_{i+1} \ldots \shuffle e_n )$$ where $\lambda' = \lambda ({s_i}_{q_i}, u, s_{i_2} )$.
		
		We have ${a_1} \cdot \lambda' ( {s_i}_{q_i}, u, s_{i_2}) \subseteq \lambda_1$ and ${a_1} \cdot \lambda' \circ \lambda_2\ \ldots \circ \lambda_j ( \{ \varepsilon \} ) = {a_1} \cdot S_i$, hence $\lambda' \circ \lambda_2 \ldots \circ \lambda_j ( \{ \varepsilon \} ) = S_i$.
		
		If we apply rule $r$ to $( c, \pi )$, we move to a configuration $( c_{1}, \pi')$, with $\pi' = {a_{2}} \ldots {a_{k}}$ and $\pi' \in ( ( {a_1} \cdot S_1 ) \cup S'_1 ) \shuffle \ldots \shuffle ( S_i ) \shuffle \ldots ( ( {a_1} \cdot S_n ) \cup S'_n ) = e_1 ( \{ \varepsilon \} ) \shuffle \ldots e_{i-1} ( \{ \varepsilon \} ) \shuffle ( \lambda' \circ \lambda_2 \ldots \circ \lambda_j ) ( \{ \varepsilon \} ) \shuffle e_{i+1} ( \{ \varepsilon \} ) \ldots \shuffle e_n ( \{ \varepsilon \} )$. Moreover, $\left| \pi' \right| < \left| \pi \right|$.
		
		If we apply the induction hypothesis, there is $c$ such that $s_{init} \xrightarrow{c}_A q_{F}$ and $( c_{1}, \pi') \leadsto_{M, \Pi} ( c, \varepsilon )$. Since $( c', \pi ) \leadsto_{M, \Pi} ( c_{1}, \pi' )$, $( c', \pi ) \leadsto_{M, \Pi} (c', \varepsilon )$.
		
		\item if $a_1= \tau$, from $c$, the automaton $M$ can either move to another configuration if two of its processes synchronize with an action $a$ or apply an internal action; we focus on the first case, the second case being similar to the previous unsynchronized action in terms of pushdown operations.
		
		For $i \in \{ 1, \ldots, n \}$, we split $e_i ( \{ \varepsilon \} )$ in three parts: the set ${a} \cdot S_i$ of path expressions starting by $a$, the set ${\overline{a}} \cdot S'_i$ of path expressions starting by $\overline{a}$, and the set $S''_i$ of path expressions starting with another symbol than $a$ or $\overline{a}$. We have $e_i ( \{ \varepsilon \} ) = ( {a} \cdot S_i ) \cup ( {\overline{a}} \cdot S'_i ) \cup S''_i$ and:
		\[\pi = \tau \cdot {a_2} \cdot \ldots {a_k} \in ( ( {a} \cdot S_1 ) \cup ( {\overline{a}} \cdot S'_1 ) \cup S''_1 ) \shuffle \ldots \]
		
		Let $i$ and $j$ be two integers such that ${a_2} \cdot \ldots {a_k} \in ( ( {a} \cdot S_1 ) \cup ( {\overline{a}} \cdot S'_1 ) \cup S''_1 ) \shuffle \ldots \shuffle \ldots ( S_i ) \shuffle \ldots ( S'_j ) \shuffle \ldots ( ( {a} \cdot S_n ) \cup ( {\overline{a}} \cdot S'_n ) \cup S''_n ))$. The symbol $\tau$ of $\pi$ comes from a shuffle of a word in ${a} \cdot S_i$, hence in $e_i ( \{ \varepsilon \} )$, and a word in ${\overline{a}} \cdot S'_j$, hence in $e_j ( \{ \varepsilon \} )$.
		
		We define $w'_i = \gamma_1 \gamma_2 \ldots \gamma_l$. We now consider the sub-sequence $s_i = s_{i_1} \linebreak \xRightarrow{( p'_i \gamma_1, \lambda_1)}_{A_{pre_\Pi^*}} s_{i_2} \ldots \xRightarrow{( \gamma_l, \lambda_l)} s_{i_{l+1}}$ of the accepting execution outlined earlier.
		
		We also define $w'_j = \alpha_1 \alpha_2 \ldots \alpha_m$ as well as the sub-sequence $s_j = s_{j_1} \linebreak \xRightarrow{( p'_j \alpha_1, \lambda'_1)}_{A_{pre_\Pi^*}} s_{j_2} \ldots \xRightarrow{( \alpha_m, \lambda'_m)} s_{j_{m+1}}$ of the accepting execution outlined earlier. 
		
		 We have $e_i = \lambda_1 \circ \lambda_2 \ldots \circ \lambda_l$. There are paths starting by $a$ in $e_i$, hence in $\lambda_1$ as well, and there is therefore a rule:
		\[r_{1} = p'_{i} \gamma _{1} \xrightarrow{a}_M q_{i} u\]
		
		By the saturation rules, from the transition $( {s_i}_{q_i}, u, s_{i_2} )$, we add a new transition $( {s_i}_{p'_i}, \gamma_1, s_{i_2} )$ labelled by $\lambda_1$ such that $a \cdot \lambda ( {s_i}_{q_i}, u, s_{i_2} ) \subseteq \lambda_1$.
		
		We have $e_j = \lambda'_1 \circ \lambda'_2 \ldots \circ \lambda'_m$. There are paths starting by $\overline{a}$ in $e_j$, hence in $\lambda'_1$ as well, and there is a rule:
		\[r_{2} = p'_{j} \alpha_{1} \xrightarrow{\overline{a}} q_{j} u'\]
		
		By the saturation rules, from the transition $( {s_j}_{q_j}, u', s_{j_2} )$, we add a new transition $( {s_j}_{p'_j}, \alpha_1, s_{j_2} )$ labelled by $\lambda'_1$ such that ${\overline{a}} \cdot \lambda ( {s_j}_{q_j}, u', s_{j_2} ) \subseteq \lambda'_1$.
		
		We define $v_1 = \gamma_2 \ldots \gamma_l$ and $v_2 = \alpha_2 \ldots \alpha_m$.
		
		The automaton $A_{pre_{\Pi}^{*}}$ then has an accepting execution labelled by $c_{1} = ( p'_{1} w'_{1} \ldots p'_{i-1} w'_{i-1} q_{i} u v_1 \ldots q_{j} u' v_{2} \ldots p'_{n} w'_{n} )$ and $( e_1 \shuffle \ldots e_{i-1} \shuffle ( \lambda \circ \lambda_2 \ldots \circ \lambda_l ) \shuffle \ldots ( \lambda' \circ \lambda'_2 \ldots \circ \lambda'_m ) \ldots \shuffle e_n )$, where $\lambda = \lambda ( {s_i}_{q_i}, u, s_{i_2})$ and $\lambda'= \lambda ( {s_j}_{q_j}, u', s_{j_2} )$.
		
		We have $a \cdot \lambda \subseteq \lambda_1$ and $a \cdot \lambda \circ \lambda_2 \ldots \circ \lambda_l ( \{ \varepsilon \} ) = a \cdot S_i$, hence $\lambda \circ \lambda_2 \ldots \circ \lambda_l ( \{ \varepsilon \} ) = S_i$. Moreover, ${\overline{a}} \cdot \lambda' \subseteq \lambda'_1$ and ${\overline{a}} \cdot \lambda' \circ \lambda'_2 \ldots \circ \lambda'_m (( \{ \varepsilon \} )) = {\overline{a}} \cdot S'_j$, hence $\lambda' \circ \lambda'_2 \ldots \circ \lambda'_m ( \{ \varepsilon \} ) = S'_j$.
		
		If we apply $r_1$ and $r_2$ in a synchronized manner to the configuration $(c, \pi)$, we move to another configuration $( c_{1}, \pi' )$, where $\pi' = {a_{2}} \ldots {a_{k}}$ and $\pi' \in ( ( {a} \cdot S_1 ) \cup ( {\overline{a}} \cdot S'_1 ) \cup S''_1 ) \shuffle \ldots ( S_i ) \shuffle \ldots ( S'_j ) \shuffle \ldots ( ( {a} \cdot S_n ) \cup ( {\overline{a}} \cdot S'_n ) \cup S''_n ) = e_1( \{ \varepsilon \} ) \shuffle \ldots e_{i-1} ( \{ \varepsilon \} ) \shuffle ( \lambda \circ \ldots \circ \lambda_l ) ( \{ \varepsilon \} ) \shuffle \ldots ( \lambda' \circ \lambda'_2 \ldots \circ \lambda'_m ) ( \{ \varepsilon \} ) \ldots \shuffle e_n ( \{ \varepsilon \} )$.
		
		Since $\left| \pi' \right| < \left| \pi \right|$, we can apply the induction hypothesis. There is a configuration $c$ such that $s_{init} \xrightarrow{c}_A q_{F}$ and $( c_{1}, \pi' ) \leadsto_{M, \Pi} ( c,( \varepsilon ))$. Since $ ( c', \pi ) \leadsto_{M, \Pi} ( c_{1}, \pi')$, it follows that $( c', \pi) \leadsto_{M, \Pi} ( c, \varepsilon )$.
	\end{itemize}
\end{description}

\section{An abstraction framework for paths}

We can't compute the exact set $Paths_M (C, C' )$, we will therefore over approximate it. To do so, we use the following mathematical framework, basing our technique on the approach presented by Bouajjani et al. in \cite{BET-popl03}.

\subsection{Abstractions and Galois connections}

Let $\mathcal{L} = (2^{Act^*}, \subseteq, \cup, \cap, \emptyset, Act^*)$ be the complete lattice of languages on $Act$.

Our abstraction of $\mathcal{L}$ requires a lattice $E = ( D, \le, \sqcup, \sqcap, \bot, \top)$, from now on called the \emph{abstract lattice}, where $D$ is a set called the abstract domain, as well as a pair of mappings $( \alpha, \beta )$ called a \emph{Galois connection}, where $\alpha : 2^{Act^*} \to D$ and $\beta : D \to 2^{Act^*}$ are such that $\forall x \in 2^{Act^*}$, $\forall y \in D$, $\alpha (x) \le y \Leftrightarrow x \subseteq \beta (y)$.

$\forall L \in \mathcal{L}$, given a Galois connection $(\alpha, \beta)$, we have $L \subseteq \beta (\alpha ( L ))$. Hence, the Galois connection can be used to overapproximate a language, such as the set of execution paths of a SDPN.

Moreover, it is easy to see that $\forall L_1, \forall L_2 \in \mathcal{L}$, $\alpha ( L_1 ) \sqcap \alpha ( L_2 ) = \bot$ if and only if $\beta (\alpha ( L )) \cap \beta (\alpha ( L )) = \emptyset$. We therefore only need to check if $\alpha ( Paths_M ( C, C' ) ) \sqcap \alpha ( \tau^* ) = \bot$. From then on, $\alpha ( Paths_M ( C, C' ) )$ will be called the \emph{abstraction} of $Paths_M ( C, C' )$, although technically the set $\beta ( \alpha ( Paths_M ( C, C' ) )$ is the actual over-approximation.

\subsection{Kleene algebras}

We want to define abstractions of $\mathcal{L}$ such that we can compute the abstract path language $\alpha ( Paths_M ( C', C ) )$, assuming the sets $C'$ and $C$ are regular. In order to do so, we consider a special class of abstractions, called \emph{Kleene abstractions}.

An idempotent semiring is a structure $K = ( A, \oplus, \odot, \overline{0}, \overline{1} )$, where $\oplus$ is an associative, commutative, and idempotent ($a \oplus a = a$) operation, and $\odot$ is an associative operation. $\overline{0}$ and $\overline{1}$ are neutral elements for $\oplus$ and $\odot$ respectively, $\overline{0}$ is an annihilator for $\odot$ ($a \odot \overline{0} = \overline{0} \odot a = \overline{0}$) and $\odot$ distributes over $\oplus$.

$K$ is an $Act$-semiring if it can be generated by $\overline{0}$, $\overline{1}$, and elements of the form $v_{a} \in A$, $\forall a \in Act$. A semiring is said to be closed if $\oplus$ can be extended to an operator over countably infinite sets while keeping the same properties as $\oplus$.

We define $a^{0} = \overline{1}$, $a^{n+1} = a \odot a^{n}$ and $a^{*} = \mathop{\bigoplus} \limits_{n \ge 0} a^{n}$. Adding the $*$ operation to an idempotent closed $Act$-semiring $K$ transforms it into a \emph{Kleene algebra}.

\subsection{Kleene abstractions}
\label{kleene_abstractions}

An abstract lattice $E = ( D, \le, \sqcup, \sqcap, \bot, \top)$ is said to be compatible with a Kleene algebra $K = ( A, \oplus, \odot, \overline{0}, \overline{1} )$ if $D = A$, $x \le y \Leftrightarrow x \oplus y = y$, $\bot = \overline{0}$ and $\sqcup = \oplus$.

A \emph{Kleene abstraction} is an abstraction such that the abstract lattice $E$ is compatible with the Kleene algebra and the Galois connection $\alpha: 2^{Act^*} \to D$ and $\beta: D\to 2^{Act^*}$ is defined by:
\[\alpha (L) = \mathop{\bigoplus} \limits_{a_{1} \ldots a_{n} \in L} v_{a_1} \odot \ldots \odot v_{a_n}\] 

\[\beta (x) = \left\{a_{1} \ldots a_{n} \in 2^{Act^*} \mid v_{a_1} \odot \ldots \odot v_{a_n} \le x\right\}\] 

Intuitively, a Kleene abstraction is such that the abstract operations $\oplus$, $\odot$, and $*$ can be matched to the union, the concatenation, and the Kleene closure of the languages of the lattice $\mathcal{L}$, $\overline{0}$ and $\overline{1}$ to the empty language and $\{ \varepsilon \}$, $v_a$ to the language $\{ a \}$, the upper bound $\top \in K$ to $Act^*$, and the operation $\sqcap$ to the intersection of languages in the lattice $\mathcal{L}$.

In order to compute $\alpha ( L )$ for a given language $L$, each word $a_{1} \ldots a_{n}$ in $L$ is matched to its abstraction $v_{a_1} \odot \ldots \odot v_{a_n}$, and we consider the sum of these abstractions.

We can check if $\alpha ( Paths_M ( C, C' ) ) \sqcap \mathop{\bigoplus} \limits_{n \geq 0} v_{\tau}^n = \bot$; if it is indeed the case, then $\beta ( \alpha ( Paths_M ( C, C' ) ) ) \cap \tau^* = \emptyset$, and since $\beta ( \alpha ( Paths_M ( C, C' ) ) )$ is an over-approximation $Paths_M ( C, C' )$, it follows that $Paths_M ( C, C' ) \cap \tau^* = \emptyset$.

A \emph{finite-chain} abstraction is an abstraction such that the lattice $( K, \oplus )$ has no infinite ascending chains. In this paper, we rely on a particular class of finite-chain abstractions, called \emph{finite-domain} abstractions, whose abstract domain $K$ is finite, such as the following examples:

\subsubsection{Prefix abstractions.}

Let $n$ be an integer and $W ( n ) = \left\{ w \in Act^* \mid \left| w \right| \le n \right\}$ be the set of words of length smaller than $n$. We define the $n$-th order \emph{prefix} abstraction $\alpha^{pref}_n$ as follows: the abstract lattice $A = 2^{W}$ is generated by the elements $v_a = \left\lbrace a \right\rbrace$, $a \in Act$; $\oplus = \cup$; $U \odot V = \{ \text{pref}_n ( uv ) \mid u \in U, v \in V \}$ where $\text{pref}_n ( w )$ stands for the prefix of $w$ of length $n$ (or lower if $w$ is of length smaller than $n$); $\overline{0} = \emptyset$; and $\overline{1} = \left\{ \varepsilon \right\}$. From there, we build an abstract lattice where $\top = W$, $\sqcap = \cap$, and $\leq = \subseteq$. This abstraction is accurate for the $n$-th first steps of a run, then approximates the other steps by $Act^*$.

We can apply a prefix abstraction of order 2 to the example shown in Figure \ref{fig:4-ex1}. For ease of representation, we show a control flow graph, although we could use the procedure outlined in section \ref{sdpn_program_model} to compute an equivalent SDPN. We also consider without loss of generality that spawns, calls, and returns are silent $\varepsilon$-transitions.

We check that, starting from an initial set of configurations $C$ with a single thread $M$ in state $m_0$, the set $C'$ where $M$ is in state $m_2$ can't be reached with regards to the strict semantics.

We have $\alpha^{pref}_2 ( Paths_M ( C, C' ) ) = \{ b, b \tau, b a, b \overline{a} \}$, $\alpha^{pref}_2 ( \tau^* ) = \{ \varepsilon, \tau, \tau \tau \}$, and $\alpha^{pref}_2 ( Paths_M ( C, C' ) ) \cap \alpha^{pref}_2 ( \tau^* ) = \emptyset$, hence, $C'$ can't be reached from $C$ with regards to the strict semantics. Intuitively, the transition labelled with $b$ in thread $M$ can't synchronize as there isn't any transition labelled with $\overline{b}$ in the whole program.

\begin{figure}
	\centering
	\begin{minipage}{.25\linewidth}
		\centering
		\begin{tikzpicture}[shorten >=1pt, node distance=2.5cm, on grid, auto] 
		\node[draw=none] (m0) {$m_0$};
		\node[draw=none] (m1) [below=of m0] {$m_1$};
		\node[draw=none] (m2) [below=of m1] {$m_2$};
		\path[->] 
		(m0) edge node {$b$} (m1)
		(m1) edge node {spawn $N$} (m2)
		(m2) edge [loop right] node {$a$} (m2);
		\end{tikzpicture}
		\textbf{Thread $M$}
	\end{minipage}
	\hspace{.01\linewidth}
	\begin{minipage}{.25\linewidth}
		\centering
		\begin{tikzpicture}[shorten >=1pt, node distance=2.5cm, on grid, auto] 
		\node[draw=none] (n0) {$n_0$};
		\node[draw=none] (n1) [below=of n0] {$n_1$};
		\node[draw=none] (n2) [below=of n1] {$n_2$};
		\path[->]
		(n0) edge node {call $F$} (n1)
		(n1) edge node {spawn $N$} (n2);
		\end{tikzpicture}
		\textbf{Thread $N$}
	\end{minipage}
	\hspace{.01\linewidth}
	\begin{minipage}{.35\linewidth}
		\centering
		\begin{tikzpicture}[shorten >=1pt, node distance=2.5cm, on grid, auto] 
		\node[draw=none] (f0) {$f_0$};
		\node[draw=none] (f1) [below right=of f0] {$f_1$};
		\node[draw=none] (f2) [below left=of f0] {$f_2$};
		\node[draw=none] (f3) [below left=of f1] {return};
		\path[->] 
		(f0) edge node {$a$} (f1)
		(f0) edge node {$\overline{a}$} (f2)
		(f1) edge node {} (f3)
		(f2) edge node {} (f3);
		\end{tikzpicture}
		\textbf{Procedure $F$}
	\end{minipage}
	\captionof{figure}{Applying a second order prefix abstraction to an example.}
	\label{fig:4-ex1}
\end{figure}
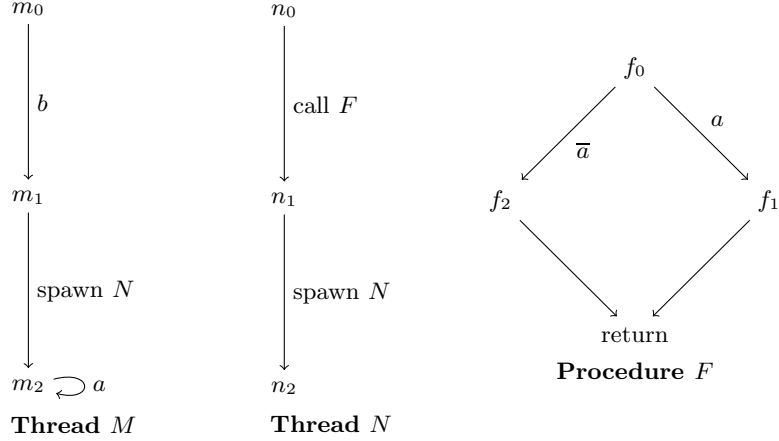

\subsubsection{Suffix abstractions.}

Let $W$ be the set of words of length smaller than $n$. We define the $n$-th order \emph{suffix} abstraction $\alpha^{suff}_n$ as follows: the abstract lattice $A = 2^{W}$ is generated by the elements $v_a = \left\lbrace a \right\rbrace$, $a \in Act$; $\oplus = \cup$; $U \odot V = \{ \text{suff}_n ( uv ) \mid u \in U, v \in V \}$ where $\text{suff}_n ( w )$ stands for the suffix of $w$ of length $n$ (or lower if $w$ is of length smaller than $n$); $\overline{0} = \emptyset$; and $\overline{1} = \left\{ \varepsilon \right\}$. From there, we build an abstract lattice where $\top = W$, $\sqcap = \cap$, and $\leq = \subseteq$. This abstraction is accurate for the $n$-th last steps of a run, then approximates the other steps by $Act^*$.

We apply a suffix abstraction of order $2$ to the example shown in Figure \ref{fig:4-ex2}. We check that, starting from an initial set of configurations $C$ with a single thread $M$ in state $m_0$, the set $C'$ where $M$ is in state $m_2$ can't be reached with regards to the strict semantics.

We have $\alpha^{suff}_2 ( Paths_M ( C, C' ) ) = \{ b, a b, \overline{a} b, \tau b \}$, $\alpha^{suff}_2 ( \tau^* ) = \{ \varepsilon, \tau, \tau \tau \}$, and $\alpha^{suff}_2 ( Paths_M ( C, C' ) ) \cap \alpha^{suff}_2 ( \tau^* ) = \emptyset$, hence, $C'$ can't be reached from $C$ with regards to the strict semantics. Intuitively, the transition labelled with $b$ in thread $M$ can't synchronize as there isn't any transition labelled with $\overline{b}$ in the whole program.

\begin{figure}
	\centering
	\begin{minipage}{.25\linewidth}
		\centering
		\begin{tikzpicture}[shorten >=1pt, node distance=2.5cm, on grid, auto] 
		\node[draw=none] (m0) {$m_0$};
		\node[draw=none] (m1) [below=of m0] {$m_1$};
		\node[draw=none] (m2) [below=of m1] {$m_2$};
		\path[->] 
		(m0) edge node {spawn $N$} (m1)
		(m1) edge [loop right] node {call $F$} (m1)
		(m1) edge node {$b$} (m2);
		\end{tikzpicture}
		\textbf{Thread $M$}
	\end{minipage}
	\hspace{.01\linewidth}
	\begin{minipage}{.25\linewidth}
		\centering
		\begin{tikzpicture}[shorten >=1pt, node distance=2.5cm, on grid, auto] 
		\node[draw=none] (n0) {$n_0$};
		\node[draw=none] (n1) [below=of n0] {$n_1$};
		\node[draw=none] (n2) [below=of n1] {$n_2$};
		\path[->]
		(n0) edge node {$\tau$} (n1)
		(n1) edge [loop right] node {call $F$} (n1)
		(n1) edge node {$a$} (n2);
		\end{tikzpicture}
		\textbf{Thread $N$}
	\end{minipage}
	\hspace{.01\linewidth}
	\begin{minipage}{.35\linewidth}
		\centering
		\begin{tikzpicture}[shorten >=1pt, node distance=2.5cm, on grid, auto] 
		\node[draw=none] (f0) {$f_0$};
		\node[draw=none] (f1) [below right=of f0] {$f_1$};
		\node[draw=none] (f2) [below left=of f0] {$f_2$};
		\node[draw=none] (f3) [below left=of f1] {return};
		\path[->] 
		(f0) edge node {$a$} (f1)
		(f0) edge node {$\overline{a}$} (f2)
		(f1) edge node {} (f3)
		(f2) edge node {} (f3);
		\end{tikzpicture}
		\textbf{Procedure $F$}
	\end{minipage}
	\captionof{figure}{Applying a second order suffix abstraction to an example.}
	\label{fig:4-ex2}
\end{figure}
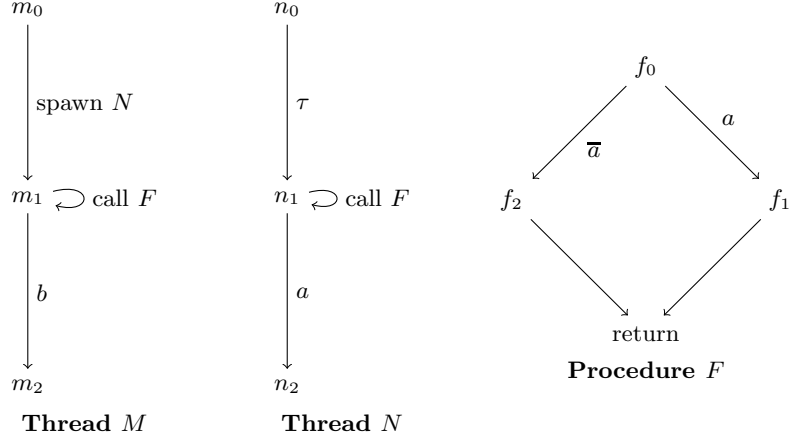

It is worth noting that the reachability problem in Example \ref{fig:4-ex2} can't be solved by a prefix abstraction, no matter its order. The reason is that $\forall n \geq 0$, there is an execution path $\tau^n b \in Paths_M ( C, C' )$, hence, $\tau^n \in \alpha^{pref}_n ( Paths_M ( C, C' ) )$. Intuitively, the two self-pointing loops of nodes $m_1$ and $n_1$ can synchronize.

Conversely, we can't use a suffix abstraction to solve the reachability problem in Example \ref{fig:4-ex1}. The reason is that $\forall n \geq 0$, there is an execution path $b \tau^n \in Paths_M ( C, C' )$, hence, $\tau^n \in \alpha^{suff}_n ( Paths_M ( C, C' ) )$. Intuitively, the self-pointing loop of node $m_2$ can synchronize with the self-spawning loop in thread $N$.

Thus, these two abstractions (prefix and suffix) complement each other.

\section{Abstracting the set of paths}

Since we can't compute the solution of the constraints outlined in \ref{labelling_constraints} on paths, our intuition now is to solve them in an abstract \emph{finite} domain defined by a Kleene abstraction where we can compute the least pre-fixpoint, in a manner similar to \cite{BET-infinity04}.

\subsection{From the language of paths to the Kleene abstraction}

We abstract the the complete lattice of languages $\mathcal{L} = (2^{Act^*}, \subseteq, \cup, \cap, \emptyset, Act^*)$ by a finite domain Kleene abstraction on an abstract lattice $E = ( D, \le, \sqcup, \sqcap, \bot, \top)$ and a Klenne algebra $K = ( A, \oplus, \odot, \overline{0}, \overline{1} )$, as defined in Section \ref{kleene_abstractions}.

Intuitively:
\begin{itemize}
	\item the set $\Pi$ is abstracted by $K$;
	\item the operator $\cdot$ is abstracted by $\odot$;
	\item the operator $\cup$ is abstracted by $\sqcup = \oplus$;
	\item the operator $\subseteq$ is abstracted by $\le$;
	\item the operator $\cap$ is abstracted by $\sqcap$;
	\item $\emptyset$ is abstracted by $\bot = \overline{0}$;
	\item $\{ \varepsilon \}$ is abstracted by $\overline{1}$;
	\item $Act^*$ is abstracted by the set of path expressions $\Pi_K$, that is, the smallest set such that:
	\begin{itemize}
		\item $\overline{1} \in \Pi_{K}$;
		
		\item if $\pi \in \Pi_{K}$, then $\forall a \in Act$, $v_a \odot \pi \in \Pi_{K}$.
	\end{itemize}
\end{itemize}

We can define $K$-configurations in $Conf_M \times \Pi_K$ and $K$-automata in a similar manner to $\Pi$-configurations and $\Pi$-automata. For a given set of configurations $C$, we introduce the set $pre_{K}^{*} ( M, C )$ of $K$-configurations $( c, \pi )$ such that $( c, \pi ) \leadsto_{M, K} ( c', \overline{1} )$ for $c' \in C$. The following property obviously holds:
\[ pre_{K}^{*} ( M, C ) = \left\{ ( c', \pi ) \mid c' \in pre^* ( M, C ), \pi \leq \alpha ( Paths_M ( \{ c' \}, C ) ) \right\} \]
The abstract path expression $\pi$ is meant to be the abstraction of an actual execution path from $c$ to $c'$.

To do so, we need to define the shuffle operation to paths expressions. However, it has to be \emph{well-defined}: given two representations $v_{a_1} \odot \ldots \odot v_{a_n} = v_{b_1} \odot \ldots \odot v_{b_m}$ of a same path expression, $\forall w \in \{v_a | a \in Act \}^*$, we must have $( v_{a_1}, \ldots, v_{a_n} ) \shuffle w = ( v_{b_1}, \ldots, v_{b_m} ) \shuffle w$.

To this end, we first inductively define a shuffle operation $\shuffle : (\{v_a | a \in Act \}^*)^2 \rightarrow K$ such that, given two sequences representing path expressions, their shuffle product is the set of all possible interleaving (with synchronization) of these sequences.

Let $w = (v_{a_1}, \ldots, v_{a_n})$ and $w' = (v_{b_1}, \ldots, v_{b_m})$ be two such sequences:
\begin{itemize}
	\item $(v_{a_1}, \ldots v_{a_n}) \shuffle (\varepsilon) = (\varepsilon) \shuffle (v_{a_1}, \ldots v_{a_n}) = \{ v_{a_1} \odot \ldots \odot v_{a_n} \}$;
	
	\item if $b_1 \neq \overline{a_1}$, then $( (v_{a_1}, \ldots, v_{a_n}) ) \shuffle ( v_{b_1}, \ldots, v_{b_m} ) = v_{a_1} \odot ( (v_{a_2}, \ldots, v_{a_n}) \shuffle ( v_{b_1}, \ldots, v_{b_m} ) ) \oplus v_{b_1} \odot ( (v_{a_1}, \ldots, v_{a_n}) \shuffle (v_{b_2}, \ldots, v_{b_m}) )$;
	
	\item if $b_1 = \overline{a_1}$, then $( (v_{a_1}, \ldots, v_{a_n}) ) \shuffle ( v_{b_1}, \ldots, v_{b_m} ) = v_{a_1} \odot ( (v_{a_2}, \ldots, v_{a_n}) \shuffle ( v_{b_1}, \ldots, v_{b_m} ) ) \oplus v_{b_1} \odot ( (v_{a_1}, \ldots v_{a_n}) \shuffle (v_{b_2}, \ldots, v_{b_m}) ) \oplus v_{\tau} \odot ( (v_{a_2}, \ldots, v_{a_n}) \shuffle (v_{b_2}, \ldots, v_{b_m}) )$; two synchronized actions $a_1$ and $\overline{a_1}$ result in an internal action $\tau$, hence, there is a component $v_\tau \odot ( w_{1} \shuffle w_{2} )$ of the shuffle product where the two paths synchronize.
\end{itemize}

We now that prove that the shuffle product is well-defined on path expressions for the prefix and suffix abstractions. We define the length $\left| \pi \right|$ of a path expression $\pi$ as the length $n$ of the smallest sequence $(v_{a_1}, \ldots, v_{a_n})$ such that to $\pi = v_{a_1} \odot \ldots \odot v_{a_n}$, length 0 meaning that $\pi = \overline{1}$.

Note that this sequence is unique for the prefix and suffix abstractions; we can therefore define a function $\theta ( \pi ) = (v_{a_1}, \ldots, v_{a_n})$ that matches to a path expression its smallest representation.

\begin{lemma}
	The shuffle product is well-defined for the prefix abstraction.
\end{lemma}

\begin{proof}
	 We will show by induction on $m + n$ that, given two sequences $(v_{a_1}, \ldots, \linebreak v_{a_n})$ and $(v_{b_1}, \ldots, v_{b_m})$, we have $(v_{a_1}, \ldots, v_{a_n}) \shuffle (v_{b_1}, \ldots, v_{b_m}) = \theta (v_{a_1} \odot \ldots \odot v_{a_n}) \shuffle \theta (v_{b_1} \odot \ldots \odot v_{b_m})$, i.e., that the shuffle of two path sequences is equal to the shuffle of their smallest representations.
	% Préfixe
	\begin{itemize}
	 	\item if $n \leq l$ and $m \leq l$, $\theta (v_{a_1} \odot \ldots \odot v_{a_n}) = (v_{a_1}, \ldots, v_{a_n})$ and $\theta (v_{b_1} \odot \ldots \odot v_{b_m}) = (v_{b_1}, \ldots, v_{b_m})$, by definition of the $l$-th prefix abstraction: indeed, the smallest representation in this abstraction of a word of length smaller than $l$ is itself.
	 	
	 	\item if $n > l$ and $m > l$, note that $\theta (v_{a_1} \odot \ldots \odot v_{a_n}) = (v_{a_1}, \ldots, v_{a_l})$ and $\theta 	(v_{b_1} \odot \ldots \odot v_{b_m}) = (v_{b_1}, \ldots, v_{b_l})$ by definition of the $l$-th prefix. If we suppose that $b_1 = \overline{a_1}$:
	 	
	 	\[ (v_{a_1}, \ldots, v_{a_n}) \shuffle (v_{b_1}, \ldots, v_{b_m}) \]
	 	\[ = v_{a_1} \odot ( (v_{a_2}, \ldots, v_{a_n}) \shuffle (v_{b_1}, \ldots, v_{b_m}) ) \oplus v_{\overline{a_1}} \odot ( (v_{a_1}, \ldots, v_{a_n}) \shuffle \] 
	 	\[ (v_{b_2}, \ldots, v_{b_m}) ) \oplus v_{\tau} \odot ( (v_{a_2}, \ldots, v_{a_n}) \shuffle (v_{b_2}, \ldots, v_{b_m}) ) \]
	 	\[ = v_{a_1} \odot ( (v_{a_2}, \ldots, v_{a_{l+1}}) \shuffle (v_{b_1}, \ldots, v_{b_l}) ) \oplus v_{\overline{a_1}} \odot ( (v_{a_1}, \ldots, v_{a_l}) \shuffle \] 
	 	\[ (v_{b_2}, \ldots, v_{b_{l+1}}) ) \oplus v_{\tau} \odot ( (v_{a_2}, \ldots, v_{a_{l+1}}) \shuffle (v_{b_2}, \ldots, v_{b_{l+1}}) ) \]
	 	
	 	if we apply the induction hypothesis. However, by definition of the prefix abstraction, given two sequences $(v_{x_1}, \ldots, v_{x_l})$ and $w$ and $\forall x_0 \in Act$, $v_{x_0} \odot ((v_{x_1}, \ldots, v_{x_l}) \shuffle w) = v_{x_0} \odot ((v_{x_1}, \ldots, v_{x_{l-1}}) \shuffle w)$. Intuitively, the symbols $v_{x_{1}}, \ldots, v_{x_{l-1}}$ have to be inserted after $v_{x_{0}}$ but before $v_{x_{l}}$, and $v_{x_{l}}$ will therefore be cut out of the prefix when we concatenate $v_{x_{0}}$. Hence:
	 	\[ (v_{a_1}, \ldots, v_{a_n}) \shuffle (v_{b_1}, \ldots, v_{b_m}) \]
	 	\[ = v_{a_1} \odot ( (v_{a_2}, \ldots, v_{a_{l}}) \shuffle (v_{b_1}, \ldots, v_{b_l}) ) \oplus v_{\overline{a_1}} \odot ( (v_{a_1}, \ldots, v_{a_l}) \shuffle \] 
	 	\[ (v_{b_2}, \ldots, v_{b_{l}}) ) \oplus v_{\tau} \odot ( (v_{a_2}, \ldots, v_{a_{l}}) \shuffle (v_{b_2}, \ldots, v_{b_{l}}) ) \]
	 	\[ = (v_{a_1}, \ldots, v_{a_l}) \shuffle (v_{b_1}, \ldots, v_{b_l}) \]
	 	The case $b_1 \neq \overline{a_1}$ is similar. Hence, the induction holds.
 	\end{itemize}
 	
 	As a consequence, the shuffle product is well-defined: if $v_{a_1} \odot \ldots \odot v_{a_n} = v_{b_1} \odot \ldots \odot v_{b_m} = \pi$, then for all sequences $w$, $(v_{a_1}, \ldots, v_{a_n}) \shuffle w = \theta(\pi) \shuffle w = (v_{b_1}, \ldots, v_{b_m}) \shuffle w$.	
\end{proof}

\begin{lemma}
	The shuffle product is well-defined for the suffix abstraction.
\end{lemma}

\begin{proof}
	We will show again by induction on $m + n$ that, given two sequences $(v_{a_1}, \ldots, v_{a_n})$ and $(v_{b_1}, \ldots, v_{b_m})$, we have $(v_{a_1}, \ldots, v_{a_n}) \shuffle (v_{b_1}, \ldots, v_{b_m}) = \theta (v_{a_1} \odot \ldots \odot v_{a_n}) \shuffle \theta (v_{b_1} \odot \ldots \odot v_{b_m})$, i.e., that the shuffle of two path sequences is equal to the shuffle of their smallest representations.
	% Suffixe
	\begin{itemize}
	 	\item if $n \leq l$ and $m \leq l$, $\theta (v_{a_1} \odot \ldots \odot v_{a_n}) = (v_{a_1}, \ldots, v_{a_n})$ and $\theta (v_{b_1} \odot \ldots \odot v_{b_m}) = (v_{b_1}, \ldots, v_{b_m})$, by definition of the $l$-th suffix abstraction: indeed, the smallest representation in this abstraction of a word of length smaller than $l$ is itself.
	 	
	 	\item if $n > l$ and $m > l$, note that $\theta (v_{a_1} \odot \ldots \odot v_{a_n}) = (v_{a_{n-l+1}}, \ldots, v_{a_l})$ and $\theta (v_{b_1} \odot \ldots \odot v_{b_m}) = (v_{b_{m-l+1}}, \ldots, v_{b_l})$ by definition of the $l$-th suffix. If we suppose that $b_1 = \overline{a_1}$:
	 	
	 	\[ (v_{a_1}, \ldots, v_{a_n}) \shuffle (v_{b_1}, \ldots, v_{b_m}) \]
	 	\[ = v_{a_1} \odot ( (v_{a_2}, \ldots, v_{a_n}) \shuffle (v_{b_1}, \ldots, v_{b_m}) ) \oplus v_{\overline{a_1}} \odot ( (v_{a_1}, \ldots, v_{a_n}) \shuffle \] 
	 	\[ (v_{b_2}, \ldots, v_{b_m}) ) \oplus v_{\tau} \odot ( (v_{a_2}, \ldots, v_{a_n}) \shuffle (v_{b_2}, \ldots, v_{b_m}) ) \]
	 	\[ = v_{a_1} \odot ( (v_{a_{n-l+1}}, \ldots, v_{a_{n}}) \shuffle (v_{b_{m-l+1}}, \ldots, v_{b_m}) ) \oplus v_{\overline{a_1}} \odot \] 
	 	\[ ( (v_{a_{n-l+1}}, \ldots, v_{a_n}) \shuffle (v_{b_{m-l+1}}, \ldots, v_{b_{m}}) ) \oplus v_{\tau} \odot ( (v_{a_{n-l+1}}, \ldots, v_{a_{n}}) \shuffle \]
	 	\[ (v_{b_{m-l+1}}, \ldots, v_{b_{m}}) ) \]
	 	
	 	if we apply the induction hypothesis. However, by definition of the suffix abstraction, given two sequences $(v_{x_1}, \ldots, v_{x_l})$ and $w$ and $\forall x_0 \in Act$, $v_{x_0} \odot ((v_{x_1}, \ldots, v_{x_l}) \shuffle w) = ((v_{x_1}, \ldots, v_{x_{l}}) \shuffle w)$. Intuitively, the symbol $v_{x_{1}}$ is concatenated to shuffled paths that are already of length greater than $l$, hence, will be cut out of the suffix. Therefore:
	 	\[ (v_{a_1}, \ldots, v_{a_n}) \shuffle (v_{b_1}, \ldots, v_{b_m}) \]
	 	\[ = (v_{a_{n-l+1}}, \ldots, v_{a_n}) \shuffle (v_{b_{m-l+1}}, \ldots, v_{b_{m}}) \]
	 	The case $b_1 \neq \overline{a_1}$ is similar. Hence, the induction holds.
	\end{itemize}
		
	As a consequence, the shuffle product is well-defined: if $v_{a_1} \odot \ldots \odot v_{a_n} = v_{b_1} \odot \ldots \odot v_{b_m} = \pi$, then for all sequences $w$, $(v_{a_1}, \ldots, v_{a_n}) \shuffle w = \theta(\pi) \shuffle w = (v_{b_1}, \ldots, v_{b_m}) \shuffle w$.	
\end{proof}

From now on, we consider that $\alpha$ is either the prefix or suffix abstraction of rank $l$.

\subsection{Computing $pre^*_K ( M, C )$}

Given a SDPN $M$ and a regular set $C$ of configurations of $M$ accepted by an $M$-automaton $A$, we want to compute a $K$-automaton $A_{pre^*_K}$ accepting $pre^*_K ( M, C )$. To this end, we will add new labels in $K^K$ to the $M$-automaton $A_{pre^*}$.

Let $Q$ be the set of states of $A$, hence, of $A_{pre^*}$ as well. We now consider the following set of constraints in the abstract domain on the labels of transitions of $A_{pre^*}$ in $S_S \times \Gamma \times S_S$:

\begin{description}
	\item[$( Y_1 )$] if $t$ belongs to $A$, then: \[ Id \le \lambda ( t ) \]
	
	\item[$( Y_2 )$] for each rule $p \gamma \xrightarrow{a} p' \gamma' \in \Delta$, for each $q \in Q$, for each $s \in S_c$: \[ v_{a} \odot \lambda ( s_{p'}, \gamma', q ) \le \lambda ( s_{p}, \gamma, q ) \]
	
	\item[$( Y_3 )$] for each rule $p \gamma \xrightarrow{a} p' \varepsilon \in \Delta$, for each $s \in S_c$: \[v_{a} \odot Id \le \lambda ( s_{p}, \gamma, s_{p'} ) \]
	
	\item[$( Y_4 )$] for each rule $p \gamma \xrightarrow{a} p' \gamma_1 \gamma_2 \in \Delta$, for each $q \in Q$, for each $s \in S_c$: \[\mathop{\bigoplus} \limits_{q' \in Q} v_{a} \odot ( \lambda ( s_{p'}, \gamma_{1}, q' ) \circ \lambda ( q', \gamma_{2}, q ) ) \le \lambda ( s_{p}, \gamma, q ) \] 
	
	\item[$( Y_5 )$] for each rule $p \gamma \xrightarrow{a} p_{2} \gamma_{2} p_{1} \gamma _{1} \in \Delta$, for each $q \in Q$, for each $s \in S_c$: \[\mathop{\bigoplus} \limits_{s'' \xrightarrow{\varepsilon}_{A_{pre^*}} s'} v_{a} \odot ( \lambda ( s_{p_2}, \gamma_{2}, s'' ) ( \overline{1} ) \shuffle \lambda (s'_{p_1}, \gamma_{1} ,q ) ) \le \lambda ( s_{p}, \gamma, q ) \]
\end{description}

Since $\alpha$ is a finite-domain abstraction, the set $K^K$ of functions in $K$ is finite as well. Let $t_1, \ldots, t_m$ be an arbitrary numbering of the transitions of $A_{pre_{K}^{*}}$ labelled with functions in the abstract domain and let $k_1, \ldots, k_n$ be an enumeration of the elements of the finite domain $K$ ($n = | K |$). The labelling constraints of section \ref{labelling_constraints} define a system of inequalities on $m * n$ variables $x_1, \ldots, x_{m n}$ such that its smallest solution is $t_1 ( k_1 ), \ldots, t_1 ( k_n )$, $t_2 ( k_1 ), \ldots, t_m ( k_n )$. It is worth noting that we can replace two different inequalities $e_1 ( x ) \leq t_i ( x )$ and $e_2 ( x ) \leq t_i ( x )$ by a single inequality $e_1 ( x ) \oplus e_2 ( x ) \leq t_i ( x )$. We therefore end up with a system of the form:
\[ f_i ( x_1, \ldots, x_{m n} ) \leq x_i, \text{ for } i = 1, \ldots, m n \]
where the functions $f_i$ are monomials in $K [ x_1, \ldots, x_{m n} ]$. Finding the least solution of this system of inequalities amounts to finding the least pre-fixpoint of the monotonic and continuous function:
\[ F ( x_1, \ldots, x_{m n} ) = ( f_1 ( x_1, \ldots , x_{m n} ), \ldots, f_{m n} ( x_1, \ldots , x_{m n} ) ) \]
By Tarski's theorem, this fixpoint exists and is equal to $\bigoplus \limits_{i \geq 0} F^i ( \overline{0} )$.

In a finite-domain, this iterative computation always terminates in a number of steps bounded by the length of the longest ascending chain in $K$, hence, $l$ for a prefix or suffix abstraction of order $l$. There are $m n$ functions $f_i$, each with a number of $\oplus$, $\odot$, and $\shuffle$ operations in $O ( |\Delta| \cdot |Q| )$. Moreover, according to \cite{BMT-concur05}, the size of the automaton $A_{pre^*}$ is $m = O ( | Q |^2 \cdot | \Delta | )$ . Each iteration step therefore features $O ( n \cdot |\Delta|^{2} \cdot | Q |^3 )$ operations, and the whole procedure, $O ( l \cdot n \cdot | \Delta |^{2} \cdot | Q |^3 )$ operations. For a prefix or suffix abstraction of order $l$, $n = 2^{|Act|^l}$, hence, a total of $O ( l \cdot 2^{|Act|^l} \cdot | \Delta |^{2} \cdot | Q |^3 )$ operations.

\subsection{Finding the abstraction}

We can compute an automaton $A_{pre_{K}^{*}}$ that accepts the set $pre_{K}^{*} ( M, C' )$. We then want to find a $K$-automaton $A'$ that accepts $pre_{K}^{*} ( M, C' ) \cap C \times \Pi_{K}$.

To do so, we define the intersection $A' = ( \Sigma, S', \delta', s'_{init}, F' )$ of the automaton $A_{pre^{*}} = ( \Sigma, S, \delta, s_{init}, F )$ with an $M$-automaton accepting $C$ called $A_1 = ( \Sigma, S_1, \delta_1, s_{1, init}, F_1 )$ accepting $C$, where $S' = S \times S_1$, $s'_{init} = ( s_{init}, s_{1, init} )$, $F = F \times F_1$, and $\delta = \{ ( q, q_1 ) \xrightarrow{a} ( q', q'_1 ) \mid q \xrightarrow{a} q' \in \delta, q \xrightarrow{a} q' \in \delta_1 \}$. Moreover, we label $A'$ with abstract functions in such a manner that $\lambda ( ( q, q_1 ), a, ( q', q'_1 ) ) = \lambda ( q , a, q' )$.

The $K)$-automaton $A'$ then obviously accepts $pre_{K}^{*} ( M, C' ) \cap C \times \Pi_{K}$. Eventually, our abstraction is $\alpha ( Paths_M ( C, C' ) ) = \mathop{\bigoplus} \{ \pi \mid ( c, \pi ) \in L_K ( A' ) \}$.

\section{Using our framework in a iterative abstraction refinement scheme}
\label{sdpn_cegar}

Following the work of Chaki et al. in \cite{CCKRT-etaps06}, we propose a semi-decision procedure that, in case of termination, allows us to answer exactly whether $Paths_M ( C, C' ) \cap \tau^* = \emptyset$.

We first model a program as a SDPN $M$, as shown in section \ref{sdpn_program_model}, its starting configurations as a regular set $C$, and a set of critical configurations whose reachability we need to study as another regular set $C'$.

We then introduce an iterative abstraction refinement scheme based on the finite-domain abstraction framework detailed previously, starting from $n = 1$.

\begin{description}
	\item[Abstraction:] we compute abstractions $\alpha ( Paths_M ( C, C' ) )$ of the set of executions paths for $\alpha = \alpha^{pref}_n$ and $\alpha = \alpha^{suff}_n$;
	
	\item[Verification:] for $\alpha = \alpha^{pref}_n$ and $\alpha = \alpha^{suff}_n$, we check if $\alpha ( Paths_M ( C, C' ) ) \sqcap \alpha ( \tau^* ) = \bot$; if it is indeed true, then we conclude that $C'$ can't be reached from $C$ using only internal or synchronized actions;
	
	\item[Counter-example validation:] if there is such a path, we then check if our abstraction introduced a spurious counter-example; this can be done in a finite number of steps by checking if this counter-example can be reached within the $n$-th first or last execution steps of the program, depending on which abstraction (prefix or suffix) provided us with a counter-example; if the counter-example is not spurious, then we conclude that $C'$ is reachable from $C$ w.r.t. the strict semantics;
	
	\item[Refinement:] if the counter-example was spurious, we go back to the first step, but use this time finite-domain abstractions of order $n + 1$.
\end{description}
If this procedure ends, we can decide the reachability problem.

\section{A case study}

We use an iterative abstraction refinement scheme to find an error in a Bluetooth driver for Windows NT. We consider here an abstracted version of a driver found in \cite{QW-pldi04} that nonetheless keeps the erroneous trace, in a manner similar to \cite{CCKRT-etaps06} and \cite{PST-cav07}. We then fix the driver by changing one of its routines, then use the abstraction scheme to prove that this new version of the driver is correct.

We model the network of processes in the driver as a SDPN. New requests for the driver are represented by thread spawns, and the driver's counter of active requests, by a counter on the stack, hence, a recursive process, making full use of our model's features.

We were able to discover the bug by applying our finite-domain abstraction in an iterative abstraction refinement scheme: we start from abstractions of order 1 and increment the order until we deduce that the erroneous configuration is reachable using a prefix abstraction of size 12. We then correct one of the program's subroutines accordingly and apply our iterative abstraction refinement scheme to prove it's now error-free.

Note that this bug was also discovered in \cite{CCKRT-etaps06,QW-pldi04,PST-cav07}. However, our approach is more complete and precise than these works: \cite{QW-pldi04} can only discover errors, whereas our scheme can also prove that the modified version of the driver is correct; \cite{CCKRT-etaps06} does not handle dynamic thread creation, and thus had to guess the number of threads for which the error arises; and \cite{PST-cav07} models thread creation as parallel calls (not as spawns), where the father process waits for its children to terminate in order to resume its execution.

\subsection{The program}

The driver consists of a certain number of processes running in parallel. Amongst these processes is an arbitrary number of requests meant to be handled by the driver. An internal counter records the number of requests currently handled by the driver: it is incremented when a request starts handling a task, and decremented once the request terminates. At any time, a special process in the driver may send a 'stop' signal; if it does, the driver switches an internal 'stopping' flag to true. The driver, however, can't stop yet and must wait until all requests have been processed.

Once the 'stopping' flag is switched on, requests may no longer perform their tasks and must instead end while decrementing the counter of active requests. Once the counter reaches zero, an internal 'stopping event' is switched to true, and the driver frees its allocated resources. If a request tries to perform a task after the resources have been released, it must abort and send an error.

Our intuition is that the 'stop' signal can interrupt a request while the latter has already started handling its task after being told the driver was still running; the driver will then free its allocated resources and let the request resume its execution, leading to an error state where the request must abort because the resources it needs are now missing.

\subsection{From the driver to the SDPN model}

We model this network of processes as a SDPN. To do so, we consider that each thread has no internal variables and a single control state, as we switch the handling of the control flow to the stack by storing the current control point of each thread on its stack. The threads can communicate and handle flags and counters by using synchronized actions: as an example, two threads can synchronize over an action $stop$ in order to switch the 'stopping' flag (represented by a control point) to true; a function {\em Increment} can also synchronize with a counter over an action $incr$ in order to increase this counter.

The driver uses the following processes:
\begin{description}
	\item[COUNTER:] this process counts the number of requests the driver receives plus the thread STOP-D; this number is set to 1 initially, is incremented when the driver receives a new request, and is decremented when a request ends;
	
	\item[STOP-D:] this process may issue a request to stop the driver at any time; it has then to wait until all the other requests have finished their work, then, when it receives a signal sent by the function {\em Decrement}, can stop the driver and free its allocated resources;
	
	\item[STOPPING-FLAG:] this process is either in state FALSE-STOP-FLAG (from then on FSF) or state TRUE-STOP-FLAG (TSF), depending on whether STOP-D is trying to stop the driver or not; it is initially in state FSF, and moves to state TSF if it receives a message from STOP-D; no new thread can enter the driver if this process is in TSF;
	
	\item[STOPPING-EVENT:] it is either in state TRUE-STOP-EVENT (TSE) or FALSE-STOP-EVENT (FSE); this process enters state TSE if the driver stops, i.e. when the number of running REQUESTs reaches 0;
	
	\item[GEN-REQ:] this process can spawn new requests as long as the driver is not trying to stop, that is, if STOPPING FLAG isn't in state TSF;
	
	\item[REQUEST:] when a new REQUEST enters the driver, it has to increment the number stored in COUNTER, perform a task, then decrement this number before exiting the driver; it uses two functions {\em Increment} and {\em Decrement} to do so.
\end{description}

If a REQUEST tries to perform its task but the allocated resources of the driver have been released, the program reaches an error state. We will check the reachability of this state.

Each process can be modelled by a SDPN as follows:

\subsubsection{The process COUNTER.}

Let $p_0$ be its unique state. The number of threads is represented by a stack. Its stack alphabet is $\{ 0, 1 \}$. Initially, the stack contains the word $1 0$, meaning that the number of request is zero and only STOP-D is running. It can then contain any word in $1^* 0$. The number of 1's in the stack corresponds to the number of running requests minus 1. The incrementation and decrementation procedures are done by receiving {\em incr} and {\em decr} actions from the functions {\em Increment} and {\em Decrement}.

COUNTER is represented by the following SDPN rules:
\begin{description}
	\item[$(r_{1a})$] $p_0 1 \xrightarrow{\overline{incr}} p_0 11$ and $(r_{1b})$ $p_0 0 \xrightarrow{\overline{incr}} p_0 1 0$; these rules increment the counter when the process is asked to do so;
	
	\item[$(r_2)$] $p_0 1 \xrightarrow{\overline{decr}} p_0 \varepsilon$; this rule decrements the counter when the process is asked to do so;
	
	\item[$(r_3a)$] $p_0 1 \xrightarrow{not-zero} p_0 1$ and $(r_3b)$ $p_0 0 \xrightarrow{is-zero} p_0 0$; these rules test whether the counter is 0 or not and send this information to other threads.
\end{description}

\subsubsection{The process STOPPING-FLAG.}

Let $p_1$ be its unique state. The process has two control points $\mbox{FSF}$ and $\mbox{TSF}$. STOPPING-FLAG is represented by the following SDPN rules:
\begin{description}
	\item[$(r_4)$] $p_1 \mbox{FSF} \xrightarrow{\overline{stop}} p_1 \mbox{TSF}$; the process receives a 'stop' request from STOP-D and changes its flag;
	
	\item[$(r_5)$] $p_1 \mbox{TSF} \xrightarrow{stopR} p_1 \mbox{TSF}$; the process sends a 'stop' message to the incoming REQUESTs;	
	
	\item[$(r_6)$] $p_1 \mbox{FSF} \xrightarrow{not-stopR} p_1 \mbox{FSF}$; the process sends a 'non-stop' request to the incoming REQUESTs.
\end{description}

\subsubsection{The process STOPPING-EVENT.}

Let $p_2$ be its unique state. The process has two control points $\mbox{FSE}$ and $\mbox{TSE}$. STOPPING-EVENT is represented by the following SDPN rules:
\begin{description}
	\item[$(r_7)$] $p_2 \mbox{FSE} \xrightarrow{\overline{has-stopped}} p_2 \mbox{TSE}$; the process receives an 'has-stopped' message and knows that the driver has stopped;
	
	\item[$(r_8)$] $p_2 \mbox{TSE} \xrightarrow{has-stopped} p_2 \mbox{TSE}$; once the driver has stopped, it keeps sending the 'has-stopped' message;
	
	\item[$(r_9)$] $p_2 \mbox{FSE} \xrightarrow{not-stopped} p_2 \mbox{FSE}$; the process sends a 'not-stopped' message if the driver is still running.	
\end{description}

\subsubsection{The process STOP-D.}

Let $p_3$ be its unique state. It has three control points $s_0$, $s_1$, and $R$, the last one standing for 'release resources'. STOP-D is represented by the following SDPN rules:

\begin{description}
	\item[$(r_{10})$] $p_3 s_0 \xrightarrow{stop} p_3 f_{Decrement} s_1$; STOP-D sends a 'stop' request to the process STOPPING-FLAG, and calls the function {\em Decrement};
	
	\item[$(r_{11})$] $p_3 s_1 \xrightarrow{\overline{has-stopped}} p_3 R$; if the driver has stopped, the allocated resources are released.
\end{description}

\subsubsection{The process REQUEST.}

The process REQUEST executes the following instructions:
\begin{itemize}
	\item it starts by calling a function {\em Increment}; this function returns -1 (stack symbol $a_{-1}$) if the STOPPING-FLAG is set to TRUE, otherwise, it increments the counter, and returns 0 (stack symbol $a_{0}$);
	\item if {\em Increment} returns 0, then REQUEST performs its task if it can assert that STOPPING-EVENT is in state FSE (i.e., that the driver is still running);
	\item it calls afterwards a function {\em Decrement} that decrements the counter; if this counter has reached 0, it sends a message to inform STOPPING-EVENT that the driver has stopped since there are no more requests running.
\end{itemize}

The process REQUEST has three control points $r_0$, $r_{Work}$, $r_{End-Work}$, and $A$, the last one standing for 'abort', and an unique state $p_4$. It can be modelled by the following SDPN rules:

\begin{description}
	\item[$(r_{12})$] $p_4 r_0 \xrightarrow{\tau} p_4 f_{Increment}$; first, the function {\em Increment} is called;
	
	\item[$(r_{13a})$] $p_4 a_0 \xrightarrow{\tau} p_4 r_{Work} r_{End-Work}$ and $(r_{13b})$ $p_4 r_{Work} \xrightarrow{\tau} \varepsilon$; if the function {\em Increment} returns 0, then REQUEST can perform its (abstracted) work;
	
	\item[$(r_{14})$] $p_4 r_{End-Work} \xrightarrow{\overline{non-stopped}} p_4 f_{Decrement}$; once the work is finished, the process checks if the driver is still running, i.e. that process STOPPING-EVENT is in FSE;
	
	\item[$(r_{15})$] $p_4 r_{End-Work}\xrightarrow{\overline{has-stopped}} p_4 A$; if it is not the case, the program has reached an erroneous configuration and aborts.
\end{description}

\subsubsection{The process GEN-REQ.}

Let $p_5$ be its unique state, and $g_0$ its unique control point. GEN-REQ is represented by the following SDPN rule:
\begin{description}
	\item[$(r_{16})$] $p_5 g_0 \xrightarrow{\overline{non-stopped}} p_4 r_0 p_5 g_0$; the process can spawn new requests as long as the driver is running.
\end{description}

\subsubsection{The function Increment.}

It has three control points $i_0$, $a_0$, and $a_{-1}$. The function {\em Increment} is represented by the following SDPN rules, as only REQUEST calls this function:

\begin{description}
	\item[$(r_{17})$] $p_4 f_{Increment} \xrightarrow{\overline{stopR}} p_4 a_{-1}$; if STOPPING-FLAG is in TSF, the function returns $-1$;
	
	\item[$(r_{18a})$] $p_4 f_{Increment} \xrightarrow{\overline{not-stopR}} p_4 i_0$ and $(r_{18b})$ $p_4 i_0 \xrightarrow{incr} p_4 a_0$; otherwise, it returns 0 and increments the counter.
\end{description}

\subsubsection{The function Decrement.}

It has two control points $d_0$ and $d_1$. The function {\em Decrement} is represented by the following SDPN rules, where $p$ stands for either $p_3$ or $p_4$, as only REQUEST and STOP-D call this function:

\begin{description}
	\item[$(r_{19})$] $p f_{Decrement} \xrightarrow{decr} p d_0$; the counter is decremented;
	
	\item[$(r_{20})$] $p d_0 \xrightarrow{\overline{not-zero}} p \varepsilon$; then, if it has not reached 0, the function terminates;
	
	\item[$(r_{21a})$] $p d_0 \xrightarrow{\overline{is-zero}} p d_1$ and $(r_{21b})$ $p d_1 \xrightarrow{has-stopped} \varepsilon$; otherwise, a message 'has-stopped' is sent to STOPPING-EVENT.
\end{description}

\medskip
We therefore model the program as a SDPN $M = ( Act, P, \Gamma, \Delta )$, with:
\begin{itemize}
	\item a set of control states $P = \{ p_0, p_1, p_2, p_3, p_4, p_5 \}$;
	
	\item a set of stack symbols $\Gamma = \{ 0, 1, \mbox{FSF}, \mbox{TSF}, \mbox{FSE}, \mbox{TSE}, s_0, s_1, R, r_0, r_{Work}, \linebreak r_{End-Work}, A, g_0, f_{Increment}, i_0, a_0, a_{-1}, f_{Decrement}, d_0, d_1 \}$;
	
	\item a set of actions $Act = \{ \tau, \} \cup L \cup \overline{L}$, where $L = \{ incr, decr, is-zero, not-zero, stopR, not-stopR, has-stopped, non-stopped, stop \}$;
	
	\item a set of transitions $\Delta = \{ r_{1a}, r_{1b}, \ldots, r_{21a}, r_{21b} \}$.
\end{itemize}

In the initial configuration, the counter is set to one, the flags in the processes STOPPING-FLAG and STOPPING-EVENT to FALSE, and all processes but REQUESTs (that will later be spawned by GEN-REQ) are running. We then need to check if the SDPN model of the program can reach with perfect synchronization a configuration where the process COUNTER has released its ressources and reached a control point $R$, while a process REQUEST has aborted its task and reached a control point $A$.

Our goal is therefore to check if, from the initial configuration $c_0 = p_0 10 \cdot p_1 FSF \cdot p_2 FSE \cdot p_3 s_0 \cdot p_5 g_0$, a configuration in the forbidden set of configurations: $$C' = ( P \Gamma^* )^* p_3 R ( P \Gamma^* )^* p_4 A \Gamma^* ( P \Gamma^* )^*$$ is reachable.

\subsection{An erroneous execution path}

We write $(r_i) \leftrightarrow (r_j)$ if we apply two rules that synchronize. The erroneous execution path is the following, starting from configuration $c_0$
\[ p_0 10 \cdot p_1 FSF \cdot p_2 FSE \cdot p_3 s_0 \cdot p_5 g_0 \]

$( r_{16} )$ GEN-REQ spawns a REQUEST;
\[ p_0 10 \cdot p_1 FSF \cdot p_2 FSE \cdot p_3 s_0 \cdot p_4 r_0 \cdot p_5 g_0 \]

$( r_{12} )$ REQUEST calls {\em Increment};
\[ p_0 10 \cdot p_1 FSF \cdot p_2 FSE \cdot p_3 s_0 \cdot p_4 f_{Increment} \cdot p_5 g_0 \]

$(r_{6}) \leftrightarrow (r_{18a})$ STOPPING-FLAG sends $not-stopR$ to {\em Increment};
\[ p_0 10 \cdot p_1 FSF \cdot p_2 FSE \cdot p_3 s_0 \cdot p_4 i_0 \cdot p_5 g_0 \]

$(r_{10}) \leftrightarrow (r_{4})$ STOP-D sends $stop$ to STOPPING-FLAG;
\[ p_0 10 \cdot p_1 TSF \cdot p_2 FSE \cdot p_3 f_{Decrement} s_1 \cdot p_4 i_0 \cdot p_5 g_0 \]

$(r_{19}) \leftrightarrow (r_{2})$ {\em Decrement} called by STOP-D sends $decr$ to COUNTER;
\[ p_0 0 \cdot p_1 TSF \cdot p_2 FSE \cdot p_3 d_0 s_1 \cdot p_4 i_0 \cdot p_5 g_0 \]

$(r_{3b}) \leftrightarrow (r_{21a})$ COUNTER sends $is-zero$ to {\em Decrement} called by STOP-D;
\[ p_0 0 \cdot p_1 TSF \cdot p_2 FSE \cdot p_3 d_1 s_1 \cdot p_4 i_0 \cdot p_5 g_0 \]

$(r_{18b}) \leftrightarrow (r_{1b})$ {\em Increment} called by REQUEST resume its execution, returns $0$, and sends $incr$ to COUNTER;
\[ p_0 10 \cdot p_1 TSF \cdot p_2 FSE \cdot p_3 d_1 s_1 \cdot p_4 a_0 \cdot p_5 g_0 \]

$(r_{21a}) \leftrightarrow (r_{8})$ {\em Decrement} called by STOP-D sends a signal $has-stopped$ to procedure STOPPING-EVENT;
\[ p_0 10 \cdot p_1 TSF \cdot p_2 TSE \cdot p_3 s_1 \cdot p_4 a_0 \cdot p_5 g_0 \]

$(r_{8}) \leftrightarrow (r_{11})$ {\em Decrement} STOPPING-EVENT sends $has-stopped$ to STOP-D that releases resources;
\[ p_0 10 \cdot p_1 TSF \cdot p_2 TSE \cdot p_3 R \cdot p_4 a_0 \cdot p_5 g_0 \]

$(r_{13a})$ REQUEST starts its task;
\[ p_0 10 \cdot p_1 TSF \cdot p_2 TSE \cdot p_3 R \cdot p_4 r_{Work} r_{End-Work} \cdot p_5 g_0 \]

$(r_{13b})$ REQUEST performs its work;
\[ p_0 10 \cdot p_1 TSF \cdot p_2 TSE \cdot p_3 R \cdot p_4 r_{End-Work} \cdot p_5 g_0 \]

$(r_{8}) \leftrightarrow (r_{15})$ STOPPING-EVENT sends $has-stopped$ to REQUEST that aborts;
\[ p_0 10 \cdot p_1 TSF \cdot p_2 TSE \cdot p_3 R \cdot p_4 A \cdot p_5 g_0 \]

This is an erroneous configuration reachable in 12 steps. We can find it using a prefix abstraction of order 12.

\section{Related work}

Wenner introduced in \cite{W-esop10} a model of \emph{weighted dynamic pushdown networks} (WDPNs), extending the work of Reps et al. on \emph{weighted pushdown systems} in \cite{RSJM-sas03} to DPNs. WDPNs share some similarities with our abstraction framework on SDPNs: each transition is labelled by a weight in a bounded idempotent semiring, these weights can be composed along execution paths, and the sum of the weights of all execution paths between two sets of configurations can be computed, provided that a simple extension of the original semiring to an abstract set of execution hedges can be found. WDPNs, however, do not feature simultaneous, synchronized actions between pushdown processes. Moreover, in order to be efficient, the extensions of the abstract domain have to be chosen on a case-by-case basis in order to label tree automata, whereas our framework works for every finite-domain abstraction and only uses finite state automata.

\emph{Multi-stack pushdown systems} (MPDSs) are pushdown systems with two or more stacks, and can be used to model synchronized parallel programs. Qadeer et al. introduced in \cite{QR-tacas05} the notion of context, that is, a part of an execution path during which only one stack of the automaton can be modified. The reachability problem within a bounded number of context switches is decidable for MPDSs. However, MPDSs have a bounded number of stacks and, unlike SDPNs, cannot therefore handle the dynamic creation of new threads.

Bouajjani et al. introduced in \cite{BESS-fsttcs05} \emph{asynchronous dynamic pushdown networks}, or ADPNs. This model extends DPNs by adding a global control state to the whole network as a mean of communication between processes; each pushdown process can then apply rules either by reading its own local state or the global state of the network. The reachability problem within a bounded number of context switches is decidable for ADPNs, where a context here stands for a part of an execution path during which transitions altering global variables are all executed by the same process. This is an under-approximation of the actual reachability problem for synchronized parallel programs, whereas we compute in this paper an over-approximation of the same problem. The former can be used to find errors in a program but, unlike the latter, does not allow one to check that a program is free from errors.

Concurrent program with recursive procedures can be modeled as a network of synchronized pushdown systems. This model, called \emph{communicating pushdown systems} (CPDSs), was introduced by Bouajjani et al. in \cite{BET-popl03}. The reachability problem being undecidable for this class of automata, a Kleene algebra framework was designed in order to find an over-approximation of the answer. Extensions of the abstraction framework of \cite{BET-popl03} were defined in \cite{BET-infinity04, T-vissas05} to compute abstractions of execution paths of multi-threaded recursive programs communicating via rendez-vous. However, unlike SDPNs, the models considered in these articles cannot describe thread spawns, where the father of a new thread can resume its execution independently of its children.

\section{Conclusion}

Our first contribution in this paper is a new pushdown system model that can handle both synchronization by rendez-vous between parallel threads and thread spawns. The reachability problem being undecidable for this class of automata, we seek to over-approximate it by abstracting the set of paths between two regular sets of configurations $C$ and $C'$.

To this end, we introduce relaxed semantics with weaker synchronization constraints and extend the Kleene algebra abstraction framework shown in \cite{BET-popl03}. We label an automaton accepting the set of predecessors of $C'$ with functions on a finite Kleene abstraction. These functions depend on a set of constraints computed according to the pushdown rules used during the saturation procedure.

This over-approximation allows us to define an iterative abstraction refinement scheme. We then apply it to find an erroneous execution trace in a Windows Bluetooth driver.

\bibliography{references_correction}

\end{document}